\documentclass[a4paper,onecolumn,11pt]{quantumarticle}
\pdfoutput=1
\usepackage[utf8]{inputenc}
\usepackage[english]{babel}
\usepackage[T1]{fontenc}
\usepackage{amsmath}
\usepackage{amssymb}
\usepackage{physics}
\usepackage{hyperref}
\usepackage[table]{xcolor}
\usepackage{colortbl}
\usepackage{booktabs}
\usepackage{enumerate}
\usepackage{enumitem}
\usepackage{amsthm}
\usepackage{caption}
\usepackage{subcaption}

\usepackage{selinput}
\SelectInputMappings{
  adieresis={ä}, 
  acute={´},
}
\usepackage{amsmath}  

\usepackage[english]{babel}
\usepackage[T1]{fontenc}
\usepackage{lipsum}
\usepackage{mdframed}
\usepackage{booktabs}
\usepackage{adjustbox}
\usepackage{longtable}
\usepackage{tikz}
\usetikzlibrary{arrows.meta,positioning,calc,fit,matrix}
\usepackage{pgfplots}
\pgfplotsset{compat=1.17}
\usepgfplotslibrary{groupplots,statistics,fillbetween}
\pgfplotsset{colormap/viridis}

\usepackage[
  backend=biber,
  style=numeric,
  sorting=none      
]{biblatex}
\usepackage{xcolor}
\usepackage[table]{xcolor}
\usepackage{booktabs}
\usepackage{listings}
\usepackage{siunitx}
\usepackage{multicol}
\lstdefinestyle{TinyPython}{
  language        = Python,
  basicstyle      = \ttfamily\fontsize{5pt}{6pt}\selectfont,  
  numbers         = left,
  numberstyle     = \ttfamily\fontsize{5pt}{6pt}\selectfont,
  stepnumber      = 1,
  numbersep       = 3pt,
  breaklines      = true,
  breakatwhitespace = true,
  showstringspaces = false,
  frame           = single,
  xleftmargin     = 0pt,
  xrightmargin    = 0pt,
  aboveskip       = 0pt,
  belowskip       = 0pt,
  tabsize         = 2,
  columns         = flexible,
  keepspaces      = true,
}

\usepackage{pifont}  
\usepackage{newunicodechar}
\newunicodechar{✖}{\ding{55}}
\usepackage[margin=1in]{geometry}   
\usepackage{setspace}              
\usepackage{amssymb} 
\usepackage{bm}
\usepackage{amsmath}
\usepackage{amsthm}  
\usepackage{tikz}            
\usepackage{quantikz} 
\usepackage{enumitem}
\usepackage{setspace}              
\usepackage{bm}
\usepackage{amsmath,amssymb,amsthm}             
\usepackage{tikz}            
\usepackage{enumitem}
\usepackage{graphicx} 
\usepackage{algpseudocode}
\usepackage{algorithm}
\usepackage{algorithmicx}
\usepackage{algpseudocode}
\usepackage{amsthm}
\usepackage{xcolor}
\usepackage[most]{tcolorbox}
\usepackage{pgfplots}
\pgfplotsset{compat=1.17}

\usepackage{booktabs}
\usepackage{float}                 
\usepackage{hyperref}              
\usepackage{url}
\usepackage{xcolor}                
\usepackage{lscape}                
\usepackage{caption}
\usepackage{subcaption}
\usepackage{etoolbox}
\AtBeginEnvironment{abstract}{\raggedright}
\usepackage{amsmath,amssymb,amsthm}
\newtheorem{theorem}{Theorem}
\newtheorem{lemma}[theorem]{Lemma}
\newtheorem{proposition}[theorem]{Proposition}
\newtheorem{corollary}[theorem]{Corollary}
\newtheorem{definition}[theorem]{Definition}
\DeclareUnicodeCharacter{0301}{'} 

\theoremstyle{remark}
\newtheorem{remark}[theorem]{Remark}

\newcommand{\OH}{\mathcal{H}_{\mathrm{OH}}}

\newcommand{\supp}{\operatorname{supp}}

\makeatletter
\g@addto@macro\normalsize{%
    \abovedisplayskip 3pt plus 1pt minus 1pt%
    \abovedisplayshortskip 3pt plus 1pt minus 1pt%
    \belowdisplayskip 3pt plus 1pt minus 1pt%
    \belowdisplayshortskip 3pt plus 1pt minus 1pt%
}
\makeatother

\def\BibTeX{{\rm B\kern-.05em{\sc i\kern-.025em b}\kern-.08em
    T\kern-.1667em\lower.7ex\hbox{E}\kern-.125emX}}

\author{Chinonso Onah}
\affiliation{Volkswagen AG, Berliner Ring 2, 38440 Wolfsburg, Germany}
\affiliation{Department of Physics, RWTH Aachen University, 52056 Aachen, Germany}
\email{chinonso.calistus.onah@volkswagen.de}

\author{Kristel Michielsen}

\affiliation{Forschungszentrum Jülich, Germany}
\affiliation{Universit\"at zu K\"oln, 50923 K\"oln, Germany}

\title{Optimal, Qubit-Efficient Quantum Vehicle Routing via Colored-Permutations}

\begin{document}
\maketitle

\begin{abstract}
\noindent
We formulate a global-position colored-permutation encoding for the capacitated vehicle routing problem. Each of the $K$ vehicles selects a disjoint partial permutation, and the sum of these $K$ color layers forms a full $n\times n$ permutation matrix that assigns every customer to exactly one visit position. On the benchmark suite of \cite{Osaba2024Qoptlib}, our end-to-end pipeline recovers the independently verified optima. This representation uses $n^2K$ binary decision variables arranged as $K$ color layers over a common permutation structure, while vehicle capacities are enforced by weighted sums over the entries of each color class—requiring no explicit load register, and hence no extra logical qubits, beyond the routing variables. In contrast, prior quantum encodings introduce an explicit capacity/load representation with $O(Q)$ or $O(\log Q)$ additional qubits for capacity $Q$ \cite{FeldEtAl2018CVRP,PalackalEtAl2023QCVRP,bentley2022q,Xie2024CVRP}. Our construction is designed to exploit the Constraint-Enhanced QAOA (CE-QAOA) framework \cite{onahce} together with its encoded-manifold analyses \cite{onahce,onahfund,onahfinite}. Building on our recent requirements analysis of quantum utility in CVRP~\cite{Onah2025Requirements}, we develop a routing optimization formulation that directly targets one of the main near-term bottlenecks, namely the additional logical-qubit cost of vehicle labels and explicit capacity constraints. Our proposal shows strong algorithmic performance in addition to qubit efficiency.  All results are fully reproducible from the companion artifact archive in Ref.~\cite{onahcvrpdata}. The feasibility oracle introduced in Alg.~\ref{alg:feasible-global-positions} may also be of independent interest as a reusable polynomial-time decoding and certification primitive for quantum and quantum-inspired routing pipelines.\footnote{\textbf{Data Availability:} All data in this paper and the Python implementation in Qiskit are made available here \url{https://doi.org/10.5281/zenodo.18798145}.
}
\end{abstract}

\section{Introduction}
\label{sec:back}

\subsection{The state of the art in quantum circuit--based routing optimization algorithms}
\label{subsec:sota-routing}

The current gate-model literature on routing optimization remains concentrated in a very small toy-instance regime. Early direct QAOA studies of the vehicle routing problem considered instances such as \((4,2)\), \((5,2)\), and \((5,3)\), where \((n,K)\) denotes the number of locations and vehicles, respectively \cite{Azad2023VRPQAOA}. Even in that line of work, the focus was primarily on demonstrating the formulation and parameter sensitivity of QAOA rather than establishing a scalable routing pipeline \cite{Azad2023VRPQAOA}. More recent gate-model work has likewise remained small and  full hardware realizations of QAOA-based routing have been demonstrated only for three-node, two-vehicle instances \cite{Azfar2025VRPHardware}, while heterogeneous-routing simulations have typically reached only a few customers, for example three customers and two trucks in a 21-qubit encoding \cite{Fitzek2024HVRP}. Likewise, several CVRP-oriented studies have avoided a direct end-to-end routing circuit altogether by decomposing the problem into a clustering phase and a set of TSP subproblems\cite{FeldEtAl2018CVRP, PalackalEtAl2023QCVRP}. In that setting, the gate-model experiments are typically restricted to \(4\)-, \(5\)-, and \(6\)-node TSPs, with only the \(4\)-node case reported to perform satisfactorily on actual hardware \cite{PalackalEtAl2023QCVRP}. Even feasibility-preserving CVRP proposals remain experimentally limited to simple instances because of simulation constraints \cite{Xie2024CVRP}.

A closely related perspective was developed in Ref.~\cite{Onah2025Requirements}, where the requirements for early quantum utility in CVRP were formulated in an encoding-agnostic way through explicit qubit- and gate-feasibility criteria. Ref.~\cite{Onah2025Requirements} showed that the central obstacle is the resource profile induced by the qubit count and circuit depth. The present paper may be read as a constructive response to the requirements problem identified there. The present work targets the qubit overhead associated with routing encodings, and in particular the overhead caused by explicit vehicle and capacity bookkeeping.  We introduce a  formulation designed so that capacity is enforced through diagonal load expressions on the native decision register, without introducing a separate load register or slack-variable block; see Proposition~\ref{prop:no-ancilla-capacity} and the surrounding discussion in Sec.~\ref{subsec:cvrp-global}. This eliminates the additional logical-qubit cost of capacity handling and allows the available qubit budget to be spent on the actual assignment-and-ordering degrees of freedom of the routing problem. In the one-hot setting, this already allows us to solve the \(K=2\) benchmark instances in Table~\ref{tab:numerical} up to the \(n=8\). A further qubit saving emerges by replacing the linear dependence on \(K\) in standard CVRP formulations with a logarithmic vehicle-label dependence inside a composite symbol register, moving quantum routing into the regime of small-scale industrial relevance discussed in Fig.~\ref{fig:routing-scale} and Sec.~\ref{sec:indu}.

The performance claims are discussed in Sec. \ref{sec:beyond-classical} and the underlying analytical mechanism reviewed in Appendix~\ref{app:fejer-mechanism}. There we summarize the positive-filtering mechanism based on the dephased Fej\'er reference model, define the mixer envelope and wrapped phase quantities entering the analysis, and restate the finite-depth and finite-shot success bound specialized to the present routing construction.

\subsection{Constraint–Enhanced QAOA}
\label{sec:ce-qaoa}
A central objective of this work is to formulate capacitated vehicle routing in a way that remains as close as possible to the native Onah--Firt--Michielsen (OFM) kernel of Ref.~\cite{onahce} to exploit the shallow depth and ancillae free constructions proposed therein. However, CVRP with  heterogeneous routing data, and in particular demand-dependent capacity information, can break the full symmetry structure that underlies the clean kernel analysis\cite{onahfund, onahfinite}. Our aim is therefore to preserve the core CE--QAOA architecture wherever possible while isolating the sources of symmetry breaking in a controlled way.  To achieve this, we need a formalism  that preserves the assignment-and-ordering structure of CVRP in the encoded-manifold. This leads to the global-position colored-permutation formalism where each global position carries one customer--vehicle label, each customer appears exactly once, and the resulting sequence of labels determines the visit order along the routing timeline. This gives a routing description whose native combinatorial structure aligns with CE--QAOA whose basic ideas we recall presently.

The quantum dynamics in CE-QAOA is rooted in the \emph{quantum alternating–operator ansatze} (``QAOA+'') paradigm\cite{Hadfield2019AOA} which acts invariantly on constraint projectors,
thereby confining evolution to structured subspaces
\cite{Hadfield2019AOA,Fuchs2022ConstrainedMixers,tsvelikhovskiy2024symmetries}. The \emph{Constraint–Enhanced QAOA} (CE–QAOA) introduced in Ref. \cite{onahce} follows this direction but makes the problem–algorithm co–design explicit. It introduces a kernel (the Onah--Firt--Michielsen kernel) designed to operate on a tensor product of blocks of one–hot manifolds and \emph{fixes} the mixer family and initial states to match the encoding as illustrated in Fig. \ref{fig:global-entangler-block-mixers}. It subsequently expands the codesign choice to include symmetries derived from the hard constraints satisfied by valid solutions\cite{onahce}; thus, realizing a strong alignment with problem structure \cite{Li2021CoDesign}. A key point to note is that the quantum dynamics we study here  need not preserve the feasible set itself, which is typically problem-specific and requires deeper and highly specialized circuit constructions with additional ancilla qubits. See, for example, feasibility-preserving mixer and initialization schemes for CVRP in~\cite{Xie2024CVRP,bentley2022q}.  We recall the formal specifications of the Onah--Firt--Michielsen (OFM) kernel in  Def. \ref{def:kernel-requirement}.

\begin{definition}[The Onah--Firt--Michielsen (OFM) kernel]
\label{def:kernel-requirement}
An optimization instance $I$ belongs to the \emph{Onah--Firt--Michielsen (OFM) kernel} if there exist
integers $n,m\in\mathbb N$ and the one-hot encoder $\mathsf E_{\mathrm{1hot}}$
that initializes the dynamics in the fixed-Hamming-weight space
\[
\OH \;=\; (\mathcal H_1)^{\otimes m},
\qquad
\mathcal H_1 \;=\; \mathrm{span}\{\ket{e_1},\dots,\ket{e_n}\}
\quad\text{(one excitation per block)}.
\]
The problem Hamiltonian splits as
\[
H_C \;=\; H_{\mathrm{pen}} \;+\; H_{\mathrm{obj}},
\]
where $H_{\mathrm{obj}}$ is diagonal in the computational basis and represents the objective. The penalty Hamiltonian $H_{\mathrm{pen}}$ is also diagonal in the computational basis and satisfies the structural conditions below.
\begin{enumerate}[label=\textup{(\alph*)}, leftmargin=2.2em]
\item \emph{Penalty structure.} $H_{\mathrm{pen}}$ is a sum of squared affine
      one-hot/degree/capacity penalties (optionally plus linear forbids) with
      integer coefficients bounded by $\mathrm{poly}(n)$. Consequently,
      $\mathrm{spec}(H_{\mathrm{pen}})\subseteq\{0,1,\dots,t_{\max}\}$ with
      $t_{\max}=O(m)=\mathrm{poly}(n)$.
\item \emph{Pattern symmetry.} $H_{\mathrm{pen}}$ is invariant under
      (i) block permutations $S_m$ and (ii) global symbol relabelings $S_n$.
      Hence the configuration space decomposes into level sets
      $L_t=\{x:\, H_{\mathrm{pen}}(x):=\langle x \mid H_{\mathrm{pen}} \mid x \rangle=t\}$ that are preserved setwise.
\item \emph{Mixer and initial state.} The block-local normalized XY mixer is
      \begin{equation}
      \label{eq:mixer}
      \widetilde H_{XY}^{(b)} \;=\;
      \frac{1}{n-1}\sum_{1\le u<v\le n}(X_u^{(b)}X_v^{(b)}+Y_u^{(b)}Y_v^{(b)}),
      \end{equation}
      with $\|\widetilde H_{XY}^{(b)}\|=O(1)$ on each block. The initial state is the uniform one-hot product
      \[
      \ket{s_0}\;=\;\ket{s_{\mathrm{blk}}}^{\otimes m},
      \qquad
      \ket{s_{\mathrm{blk}}}\;=\;\frac{1}{\sqrt n}\sum_{r=1}^{n} \ket{e_r}
      \quad\text{(a $W_n$ state per block)}.
      \]
\end{enumerate}
\end{definition}

Consequently, for an instance in the kernel with $m$ blocks of local dimension $n$, the initial state is
\[
  \ket{s_0}\;=\;\ket{s_{\mathrm{blk}}}^{\otimes m},
\]
and the mixer unitary is
\[
  U_M(\beta)\;=\;\bigotimes_{b=1}^{m}\exp\!\bigl(-i\beta\,\widetilde H_{XY}^{(b)}\bigr).
\]
A depth-$p$ CE--QAOA stack is
\[
  \ket{\psi_p(\vec\gamma,\vec\beta)}
  \;=\;
  \Bigl(\prod_{\ell=1}^{p} U_M(\beta_\ell)\,e^{-i\gamma_\ell H_C}\Bigr)\ket{s_0},
  \qquad
  \vec\gamma=(\gamma_1,\dots,\gamma_p),\;
  \vec\beta=(\beta_1,\dots,\beta_p).
\]
Because $U_M$ preserves the one-hot sector and $H_C$ is diagonal, each layer maps the encoded manifold $\OH=(\mathcal H_1)^{\otimes m}$ to itself.


\paragraph{Two routing use cases and notation.}
The goal of this paper is to study two routing use cases---(A) the capacitated vehicle routing problem (CVRP) and (B) the pickup-and-delivery problem (PDP)---and develop an OFM kernel realization for each within a unified ``colored permutation'' routing representation. Accordingly, we present two routing formulations and their associated Hamiltonian constructions. To keep notation unambiguous when both constructions are referenced side-by-side, we label all CVRP Hamiltonian terms with a superscript \(A\) and all PDP terms with a superscript \(B\) (e.g., \(H_{\mathrm{obj}}^{A}\), \(H_{\mathrm{pen}}^{A}\) versus \(H_{\mathrm{obj}}^{B}\), \(H_{\mathrm{pen}}^{B}\)). The specification for PDP is given in App. \ref{subsec:pdp-global}. Unless explicitly stated otherwise, our numerical study reports results for the CVRP instances only. The superscript \(A\) therefore persists throughout the Hamiltonian and circuit specification in the main text.

\newcommand{\NB}{3}         
\newcommand{\Nk}{4}         
\newcommand{\TotWires}{12}  

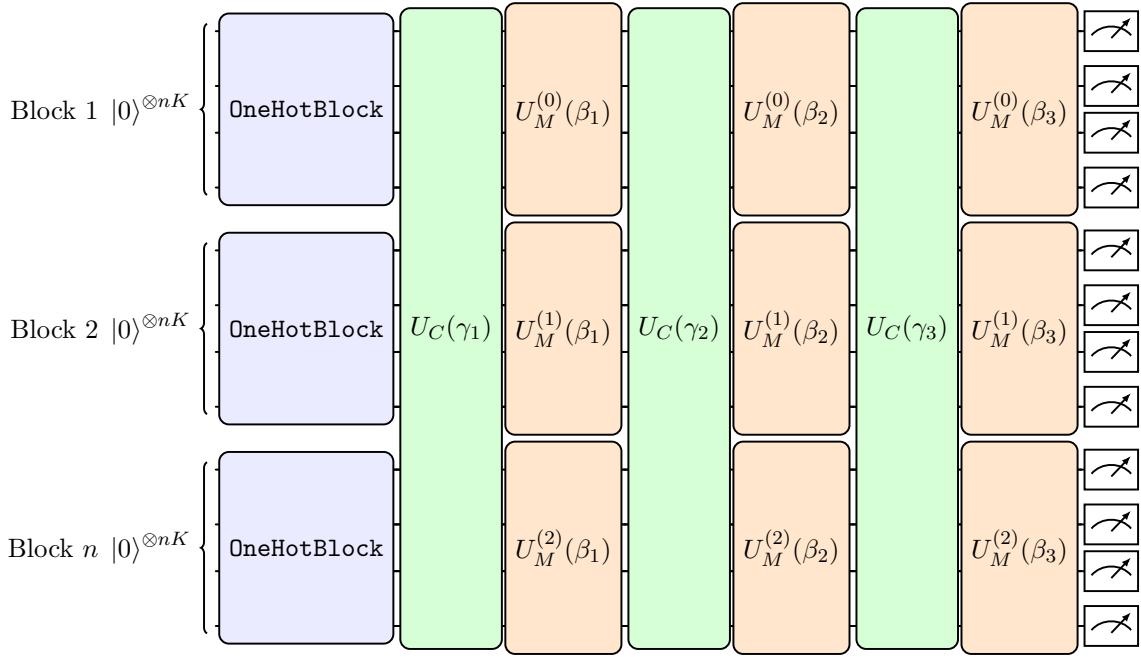
\begin{figure}[htbp]
  \centering
  \resizebox{\textwidth}{!}{%
  \begin{quantikz}[row sep=0.1cm, column sep=0.05cm]
  \lstick[wires=4]{$\text{Block }1\; \ket{0}^{\otimes nK}$}
      & \gate[wires=4,style={rounded corners,fill=blue!8}]{\texttt{OneHotBlock}}
      & \qw
      & \gate[wires=12,style={rounded corners,fill=green!16}]{U_C(\gamma_1)}
      & \gate[wires=4,style={rounded corners,fill=orange!20}]{U_M^{(0)}(\beta_1)}
      & \qw
      & \gate[wires=12,style={rounded corners,fill=green!16}]{U_C(\gamma_2)}
      & \gate[wires=4,style={rounded corners,fill=orange!20}]{U_M^{(0)}(\beta_2)}
      & \qw
      & \gate[wires=12,style={rounded corners,fill=green!16}]{U_C(\gamma_3)}
      & \gate[wires=4,style={rounded corners,fill=orange!20}]{U_M^{(0)}(\beta_3)}
      & \qw
      & \meter{} \\
  & & \qw
    & \qw & \qw & \qw
    & \qw & \qw & \qw
    & \qw & \qw & \qw
    & \meter{} \\
  & & \qw
    & \qw & \qw & \qw
    & \qw & \qw & \qw
    & \qw & \qw & \qw
    & \meter{} \\
  & & \qw
    & \qw & \qw & \qw
    & \qw & \qw & \qw
    & \qw & \qw & \qw
    & \meter{} \\
  \lstick[wires=4]{$\text{Block }2\; \ket{0}^{\otimes nK}$}
      & \gate[wires=4,style={rounded corners,fill=blue!8}]{\texttt{OneHotBlock}}
      & \qw
      & \qw
      & \gate[wires=4,style={rounded corners,fill=orange!20}]{U_M^{(1)}(\beta_1)}
      & \qw
      & \qw
      & \gate[wires=4,style={rounded corners,fill=orange!20}]{U_M^{(1)}(\beta_2)}
      & \qw
      & \qw
      & \gate[wires=4,style={rounded corners,fill=orange!20}]{U_M^{(1)}(\beta_3)}
      & \qw
      & \meter{} \\
  & & \qw
    & \qw & \qw & \qw
    & \qw & \qw & \qw
    & \qw & \qw & \qw
    & \meter{} \\
  & & \qw
    & \qw & \qw & \qw
    & \qw & \qw & \qw
    & \qw & \qw & \qw
    & \meter{} \\
  & & \qw
    & \qw & \qw & \qw
    & \qw & \qw & \qw
    & \qw & \qw & \qw
    & \meter{} \\
  \lstick[wires=4]{$\text{Block } n\; \ket{0}^{\otimes nK}$}
      & \gate[wires=4,style={rounded corners,fill=blue!8}]{\texttt{OneHotBlock}}
      & \qw
      & \qw
      & \gate[wires=4,style={rounded corners,fill=orange!20}]{U_M^{(2)}(\beta_1)}
      & \qw
      & \qw
      & \gate[wires=4,style={rounded corners,fill=orange!20}]{U_M^{(2)}(\beta_2)}
      & \qw
      & \qw
      & \gate[wires=4,style={rounded corners,fill=orange!20}]{U_M^{(2)}(\beta_3)}
      & \qw
      & \meter{} \\
  & & \qw
    & \qw & \qw & \qw
    & \qw & \qw & \qw
    & \qw & \qw & \qw
    & \meter{} \\
  & & \qw
    & \qw & \qw & \qw
    & \qw & \qw & \qw
    & \qw & \qw & \qw
    & \meter{} \\
  & & \qw
    & \qw & \qw & \qw
    & \qw & \qw & \qw
    & \qw & \qw & \qw
    & \meter{} \\
  \end{quantikz}}
  \caption{Global-position CE-QAOA with \(\NB\) blocks, each of size \(nK=\Nk\) wires.
  The diagonal entangler \(e^{-i\gamma_\ell H_C}\) spans all \(\NB\cdot \Nk\) wires (green), while
  the mixers \(U_M^{(j)}(\beta_\ell)\) (orange) act \emph{independently} within each block \(j\).
  Three layers (\(\ell=1,2,3\)) are shown.}
  \label{fig:global-entangler-block-mixers}
\end{figure}

\section{Kernelizing Vehicle Routing Problem}
\label{subsec:cvrp-global}

\subsection{Encoding All Valid Routes as Coloured Permutations}
\paragraph{Setup.}
We adopt the \emph{global-position} encoding throughout. Let $n$ be the number of items (customers for CVRP or pickup--delivery tours for PDP) and $K$ the number of vehicles. Then  we have $n$ customers (items) indexed by $i\in[n]$, $K$ trucks indexed by $k\in[K]$, and
$n$ \emph{global positions} indexed by $j\in[n]$. Consider items $i\in[n]$ with demands $d_i>0$, trucks $k\in[K]$ with capacities $Q_k>0$. Each truck $k$ has a depot $\mathtt{dep}_k$.
Let $w_{a,b}\ge 0$ be a (possibly asymmetric) metric on $\{\mathtt{dep}_1,\dots,\mathtt{dep}_K\}\cup[n]$. A \emph{configuration} says: “at position $j$, place item $i$ and assign it to truck $k$.”
So each position $j$ chooses \emph{one} pair $(i,k)$ out of the local alphabet
\[
\mathcal A\;=\;\{(i,k): i\in[n],\,k\in[K]\},
\qquad
|\mathcal A|=nK.
\]

Introduce binary variables
\[
x_{i,j,k}\in\{0,1\}\quad (i\in[n],\,j\in[n],\,k\in[K]),
\]
meaning $x_{i,j,k}=1$ iff “at position $j$ we place item $i$ on truck $k$.”
Stacking these gives an $n\times n\times K$ $0/1$ tensor $X=(x_{i,j,k})$,
or equivalently a length-$n^2K$ bitstring in lexicographic order. 
Saying “each block $j$ chooses exactly one symbol $(i,k)$” is exactly
\begin{equation}
\label{eq:block-onehot}
\forall\,j\in[n]:\qquad
\sum_{i=1}^n\;\sum_{k=1}^K x_{i,j,k}\;=\;1.
\end{equation}
In the quantum encoding, this is the one-hot subspace $\mathcal H_1$ on each block. We also want “each item $i$ is used exactly once somewhere (by some truck, at some position)”:
\begin{equation}
\label{eq:item-once}
\forall\,i\in[n]:\qquad
\sum_{j=1}^n\;\sum_{k=1}^K x_{i,j,k}\;=\;1.
\end{equation}

Given $X=(x_{i,j,k})$, define the $n\times n$ matrix
\begin{equation}
\label{eq:sum-over-k}
P\;\;:=\;\;\bigl(p_{i,j}\bigr)_{i,j=1}^n,
\qquad
p_{i,j}\;:=\;\sum_{k=1}^K x_{i,j,k}.
\end{equation}
Intuitively, $p_{i,j}=1$ means “item $i$ is placed at position $j$ (by \emph{some} truck),”
while $p_{i,j}=0$ means “no truck places item $i$ at position $j$.”

\medskip
We now prove that the feasible bitstrings under \eqref{eq:block-onehot} and \eqref{eq:item-once}
are exactly those whose $K$ slices sum to a \emph{permutation matrix} over $(i,j)$,
i.e.\ $P$ has one $1$ in each row and each column (and $0$ elsewhere).
Equivalently, a feasible $X$ is a \emph{colored decomposition} of a permutation matrix:
the $K$ color layers $\{X^{(k)}\}_{k=1}^K$ (where $X^{(k)}=(x_{i,j,k})_{i,j}$)
are disjoint $0/1$ matrices that add up to $P$.

\begin{lemma}[One-hot $+$ item-once $\Longrightarrow$ $P$ is a permutation matrix]
\label{lem:perm}
Suppose $X=(x_{i,j,k})$ satisfies \eqref{eq:block-onehot} and \eqref{eq:item-once}.
Then $P$ in \eqref{eq:sum-over-k} is a permutation matrix: for each column $j$,
$\sum_i p_{i,j}=1$, and for each row $i$, $\sum_j p_{i,j}=1$, with $p_{i,j}\in\{0,1\}$.
\end{lemma}

\begin{proof}
Fix a position $j$. By \eqref{eq:block-onehot},
\[
\sum_{i=1}^n \sum_{k=1}^K x_{i,j,k} \;=\;1.
\]
Regrouping by $i$,
\[
\sum_{i=1}^n \Bigl(\sum_{k=1}^K x_{i,j,k}\Bigr) \;=\; \sum_{i=1}^n p_{i,j} \;=\; 1,
\]
so column $j$ of $P$ has total $1$. Because all $x_{i,j,k}\in\{0,1\}$ and at a fixed $j$
there is exactly one triple $(i,k)$ with $x_{i,j,k}=1$, it follows that exactly one $i$
has $p_{i,j}=1$ and all other $p_{i,j'}=0$. Hence column $j$ is one-hot.

Similarly, fix an item $i$. By \eqref{eq:item-once},
\[
\sum_{j=1}^n \sum_{k=1}^K x_{i,j,k} \;=\; 1,
\]
so
\[
\sum_{j=1}^n p_{i,j} \;=\;1.
\]
Again, by integrality of $x_{i,j,k}$’s, exactly one column $j$ satisfies $p_{i,j}=1$.
Thus row $i$ is one-hot. Therefore $P$ is a permutation matrix.
\end{proof}

\begin{lemma}[Colored completion: any permutation can be lifted to $K$ layers]
\label{lem:colored}
Conversely, let $P$ be any permutation matrix on $[n]\times[n]$, and let
$k_1,\dots,k_n\in[K]$ be any choice of trucks for the $n$ positions.
Define $X=(x_{i,j,k})$ by
\[
x_{i,j,k}\;=\;
\begin{cases}
1,&\text{if } P_{i,j}=1 \text{ and } k=k_j,\\
0,&\text{otherwise.}
\end{cases}
\]
Then $X$ satisfies \eqref{eq:block-onehot} and \eqref{eq:item-once}, and its $K$ slices
sum to $P$ via \eqref{eq:sum-over-k}.
\end{lemma}

\begin{proof}
Because $P$ has exactly one $1$ in each column $j$, there is a unique $i_j$ with $P_{i_j,j}=1$.
By construction, at position $j$ we set exactly one variable to $1$, namely
$x_{i_j,j,k_j}=1$, so \eqref{eq:block-onehot} holds.
Because $P$ has exactly one $1$ in each row $i$, there is a unique column $j(i)$ with $P_{i,j(i)}=1$,
hence exactly one variable $x_{i,j(i),k_{j(i)}}$ is $1$ among $\{x_{i,j,k}\}_{j,k}$, so
\eqref{eq:item-once} holds. Finally,
\[
p_{i,j}=\sum_k x_{i,j,k}
=
\begin{cases}
1,&\text{if }P_{i,j}=1,\\
0,&\text{otherwise,}
\end{cases}
\]
so $\sum_k X^{(k)}=P$.
\end{proof}

\begin{theorem}[Characterization of valid routes as colored permutations]
\label{thm:colored-perms}
A tensor \(X=(x_{i,j,k})\in\{0,1\}^{n\times n\times K}\) satisfies
\eqref{eq:block-onehot} and \eqref{eq:item-once}
if and only if the matrices \(X^{(k)}=(x_{i,j,k})_{i,j}\) have pairwise disjoint supports and
\[
P := \sum_{k=1}^K X^{(k)}
\]
is a permutation matrix on \([n]\times[n]\).
In that case, each \(X^{(k)}\) is automatically a partial permutation matrix.
\end{theorem}

\begin{proof}
\((\Rightarrow)\) By Lemma~\ref{lem:perm}, the matrix
\(
P=\sum_k X^{(k)}
\)
is a permutation matrix. Since every block \(j\) is one-hot, there is exactly one active pair \((i,k)\) at each column \(j\), so the supports of the \(X^{(k)}\) are pairwise disjoint.

\((\Leftarrow)\) Suppose the supports of the \(X^{(k)}\) are pairwise disjoint and
\(
P=\sum_k X^{(k)}
\)
is a permutation matrix. Fix a column \(j\). Since \(P\) is a permutation matrix, exactly one row \(i\) has \(P_{i,j}=1\). Because the supports of the \(X^{(k)}\) are disjoint and their sum equals \(P\), there is exactly one vehicle label \(k\) for which \(x_{i,j,k}=1\), and all other entries in column \(j\) vanish. Hence
\[
\sum_{i=1}^n\sum_{k=1}^K x_{i,j,k}=1,
\]
so \eqref{eq:block-onehot} holds. The same argument applied row-wise shows that for each customer \(i\),
\[
\sum_{j=1}^n\sum_{k=1}^K x_{i,j,k}=1,
\]
so \eqref{eq:item-once} holds. Since each \(X^{(k)}\) is supported inside a permutation matrix, it has at most one \(1\) in each row and column, hence each \(X^{(k)}\) is a partial permutation matrix.
\end{proof}

Figure~\ref{fig:colored-decomposition} gives an explicit example of such a
colored decomposition, where a permutation matrix is written as the sum of two disjoint partial permutation layers.

\subsection{Constraint Enhanced Quantum Vehicle Routing Optimization}
Denote the rank-$1$ projector for symbol $(i,k)$ on block $j$ by
$X^{(j)}_{i,k}:=\ket{(i,k)}\!\bra{(i,k)}$ (acting as identity elsewhere). One-hot per block is enforced by the mixer/encoder, so we penalize only the global uniqueness equality. The constraint that \emph{each customer is visited exactly once} becomes
\begin{align}
\label{eq:cvrp-A-assign}
H_{\mathrm{once}}^{\textsc{A}}
&:=\ \lambda_{\mathrm{once}} \sum_{i=1}^{n}
\Bigl(\ \sum_{j=1}^{n}\sum_{k=1}^{K} X^{(j)}_{i,k} \;-\; 1\ \Bigr)^{2}.
\end{align}

\begin{proposition}[Capacity enforcement requires \emph{no} ancilla qubits]
\label{prop:no-ancilla-capacity}
Under the global-position one-hot encoding with decision projectors $X^{(j)}_{i,k}$ on the
native register, define the operator-valued load of vehicle $k$ by
\[
\widehat D_k \;:=\; \sum_{j=1}^{n}\sum_{i=1}^{n} d_i\,X^{(j)}_{i,k}.
\]
The hinge-square capacity penalty
\begin{equation}
\label{eq:cvrp-A-cap}
H_{\mathrm{cap}}^{\textsc{A}}
\;:=\;
\lambda_{\mathrm{cap}}\sum_{k=1}^{K}
\Bigl(\bigl[\widehat D_k - Q_k\bigr]_+\Bigr)^{2}
\end{equation}
is diagonal in the computational basis and acts only on the decision qubits.
Therefore, enforcing capacities \emph{adds zero logical ancilla qubits}.
\end{proposition}

\begin{proof}
Each $X^{(j)}_{i,k}$ is a rank-$1$ projector onto the basis state where global position $j$
is occupied by customer $i$ on vehicle $k$. In particular, each $X^{(j)}_{i,k}$ is diagonal in
the computational basis, and all such projectors commute.
Hence \(\widehat D_k=\sum_{j,i} d_i X^{(j)}_{i,k}\) is also diagonal.

For any computational basis state \(|x\rangle\) encoding decision bits
\(\{x_{i,j,k}\}\in\{0,1\}\), the operator \(\widehat D_k\) has eigenvalue
\[
D_k(x) \;=\; \sum_{j=1}^{n}\sum_{i=1}^{n} d_i\,x_{i,j,k},
\]
so the capacity penalty acts entrywise as
\[
H_{\mathrm{cap}}^{\textsc{A}}\,|x\rangle
\;=\;
\lambda_{\mathrm{cap}}\sum_{k=1}^{K}\max\{0,\,D_k(x)-Q_k\}^{2}\;|x\rangle.
\]
Therefore \(H_{\mathrm{cap}}^{\textsc{A}}\) is diagonal and depends only on the decision variables.

Since \(H_{\mathrm{cap}}^{\textsc{A}}\) is diagonal on, and supported entirely within, the native decision
register, the phase separator \(e^{-i\gamma H_{\mathrm{cap}}^{\textsc{A}}}\) acts on the decision qubits alone.
Thus capacity enforcement introduces no additional logical ancilla qubits. Any scratch qubits used by a particular synthesis strategy can be uncomputed and are not part of the encoding.
\end{proof}

This is a statement about the logical encoding. A particular gate-level
synthesis of the diagonal phase separator may use temporary workspace or
introduce high-order Pauli decompositions, but such scratch qubits are not
additional problem variables and can be uncomputed.

\begin{remark}[Quadratic surrogate in the balanced/uniform-capacity regime]
\label{rem:cap-quadratic-surrogate}
If $Q_k\equiv Q$ and routes are expected to be approximately balanced (so $D_k\approx Q$),
one may use the smooth quadratic deviation penalty
\begin{equation}
\label{eq:cap-soft-surrogate}
\widetilde H_{\mathrm{cap}}
\;:=\;
\lambda_{\mathrm{cap}}\sum_{k=1}^{K}\bigl(\widehat D_k - Q\bigr)^2,
\end{equation}
which is also diagonal and ancilla-free.
Unlike the hinge-square, $\widetilde H_{\mathrm{cap}}$ penalizes \emph{both} overload and underload,
so it behaves as a load-balancing regularizer rather than a pure “no-overload” soft constraint.
In particular, on the manifold where every customer is assigned exactly once, the total demand
$\sum_k \widehat D_k$ is fixed (a constant times the identity), so \eqref{eq:cap-soft-surrogate}
differs from $\lambda_{\mathrm{cap}}\sum_k \widehat D_k^2$ only by an additive constant and a
(global) constant shift in energy. In this balanced regime, \eqref{eq:cap-soft-surrogate} often
serves as a convenient diagonal surrogate, while \eqref{eq:cvrp-A-cap} remains the
faithful soft-penalty formulation for strict capacity feasibility ($D_k\le Q_k$).
\end{remark}

In CVRP instances with heterogeneous demands \(\{d_i\}\), capacity expressions such as
\(\widehat D_k=\sum_{j,i} d_i\,X^{(j)}_{i,k}\) are diagonal but carry instance-specific
coefficients and therefore need not be invariant under the full global customer
relabeling \(S_n\) appearing in Def.~\ref{def:kernel-requirement}(b).
Accordingly, in the present construction we take \(H_{\mathrm{pen}}\) to include only the purely combinatorial, label-symmetric
feasibility terms, such as the global uniqueness constraint \(H_{\mathrm{once}}^{\textsc{A}}\) and any optional one-hot/forbid terms that are independent of \(\{d_i\}\) and \(\{Q_k\}\). Capacity is then enforced \emph{outside} the symmetric penalty sector, either by classical feasibility filtering at decode time (rejecting samples with
\(D_k(x)>Q_k\)) or by adding an ancilla-free diagonal bias term involving \(\widehat D_k\) to
the routing Hamiltonian (for example \(\widetilde H_{\mathrm{cap}}\) from
Rem.~\ref{rem:cap-quadratic-surrogate}, or the hinge-square penalty from Prop.~\ref{prop:no-ancilla-capacity}). 

\paragraph{Objective (travel cost on the global timeline).}
Define the edge-cost between consecutive global positions by
\[
  W_{(i,k)\to(i',k')}
  \;:=\;
  \begin{cases}
    w_{i,i'} & \text{if } k=k',\\[2pt]
    w_{i,\mathtt{dep}_k} + w_{\mathtt{dep}_{k'},\,i'} & \text{if } k\neq k'.
  \end{cases}
\]
Then
\begin{align}
\label{eq:cvrp-A-obj}
H_{\mathrm{obj}}^{\textsc{A}}
\;:=\;
\lambda_{\mathrm{obj}}\Biggl[
&\sum_{j=1}^{n-1}\sum_{i,i'=1}^{n}\sum_{k,k'=1}^{K}
  W_{(i,k)\to(i',k')} \; X^{(j)}_{i,k}\,X^{(j+1)}_{i',k'}\\[-1pt]
&\quad + \sum_{i=1}^{n}\sum_{k=1}^{K} w_{\mathtt{dep}_k,\,i}\; X^{(1)}_{i,k}
\quad + \sum_{i=1}^{n}\sum_{k=1}^{K} w_{i,\,\mathtt{dep}_k}\; X^{(n)}_{i,k}
\Biggr].\nonumber
\end{align}
Within-truck adjacencies pay $w_{i,i'}$; cross-truck breaks pay depot-close $+$ depot-start; endpoints pay depot edges. The total diagonal cost Hamiltonian becomes 

\begin{equation}
\label{eq:cvrp-A-total}
H_C^{\textsc{A}} \;:=\; H_{\mathrm{once}}^{\textsc{A}} + H_{\mathrm{cap}}^{\textsc{A}} + H_{\mathrm{obj}}^{\textsc{A}}.
\end{equation}

The CE--QAOA stack is
$\ket{\psi_p(\vec\gamma,\vec\beta)}
=\bigl(\prod_{\ell=1}^p U_M(\beta_\ell)\,e^{-i\gamma_\ell H_C^{\textsc{A}}}\bigr)\ket{s_0}$.
Because $H_C^{\textsc{A}}$ couples blocks (assignment, timeline edges), $U_C(\gamma)$
is entangling. A measured bitstring selects, for each $j$, exactly one $(i,k)$. The subsequence of positions colored by $k$ yields truck $k$’s route; depot edges insert automatically at route boundaries via \eqref{eq:cvrp-A-obj}.
Route lengths emerge from the count of positions assigned to each $k$.  Under \eqref{eq:block-onehot}–\eqref{eq:item-once}, the marginals
$X_{i,j}:=\sum_{k=1}^{K}x_{i,j,k}$ form a permutation matrix on $[n]\times[n]$
(“colored permutation” view), while the $k$-index partitions items across vehicles.


\subsection{Binary compression of local symbol blocks}
\label{subsec:binary-compression}

Let
\[
S := nK,
\qquad
q := \lceil \log_2 S \rceil,
\qquad
\mathcal V := \{0,1,\dots,S-1\}.
\]
Each global-position block carries one composite symbol \((i,k)\in[n]\times[K]\).
Fix the local symbol index map
\begin{equation}
\label{eq:symbol-index-map}
\nu(i,k) := (k-1)n + (i-1)\in\mathcal V
\qquad\text{(1-based indexing).}
\end{equation}
Equivalently, in zero-based implementation indexing one writes
\[
s = i + nk,
\qquad i\in\{0,\dots,n-1\},\quad k\in\{0,\dots,K-1\}.
\]

Let \(\mathrm{bin}_q(s)\in\{0,1\}^q\) denote the \(q\)-bit binary encoding of \(s\).
The blockwise compression is the isometry
\begin{equation}
\label{eq:block-compression-isometry}
\mathcal C_{\mathrm{blk}}
:=
\sum_{s=0}^{S-1}
\ket{\mathrm{bin}_q(s)}\!\bra{e_s},
\end{equation}
from the one-hot symbol space \(\mathrm{span}\{\ket{e_0},\dots,\ket{e_{S-1}}\}\)
into the valid binary code space
\[
\mathcal K_{\mathrm{val}}
:=
\mathrm{span}\{\ket{\mathrm{bin}_q(s)}: 0\le s\le S-1\}
\subseteq (\mathbb C^2)^{\otimes q}.
\]
For the full \(n\)-block register we define
\[
\mathcal C := \mathcal C_{\mathrm{blk}}^{\otimes n}.
\]

Classically, this induces a direct blockwise compression of one-hot bitstrings into binary strings and an inverse translation back to one-hot form. The latter allows the existing feasibility oracle Alg.~\ref{alg:feasible-global-positions} to be reused without modification. Thus the binary layer changes the execution register, but not the semantics of the routing decoder or the feasibility criteria used by the hybrid pipeline. In particular, in Algorithms~\ref{alg:compress-onehot-to-binary} and \ref{alg:decode-binary-to-onehot} we show that the binary register can be treated as a compressed execution layer on top of the original one-hot routing formulation. The quantum device samples in the binary register, the inverse translation lifts the sample back to the one-hot tensor representation, and the current feasibility oracle
Alg.~\ref{alg:feasible-global-positions} remains the final certification backend.

\begin{algorithm}[H]
\caption{\textsc{CompressOneHotToBinaryBlocks}}
\label{alg:compress-onehot-to-binary}
\begin{algorithmic}[1]
\Require One-hot bitstring \(x\in\{0,1\}^{nS}\), with \(S=nK\) and \(q=\lceil\log_2 S\rceil\).
\Ensure Binary-compressed bitstring \(y\in\{0,1\}^{nq}\), or \textbf{fail} if a block is not one-hot.
\State initialize \(y\gets 0^{nq}\)
\For{\(j\gets 0\) \textbf{to} \(n-1\)}
    \State \(L \gets jS\)
    \State \(\texttt{ones}\gets 0,\quad s^\star\gets -1\)
    \For{\(s\gets 0\) \textbf{to} \(S-1\)}
        \If{\(x[L+s]=1\)}
            \State \(\texttt{ones}\gets \texttt{ones}+1,\quad s^\star\gets s\)
        \EndIf
    \EndFor
    \If{\(\texttt{ones}\neq 1\)}
        \State \Return \textbf{fail}
    \EndIf
    \State write the \(q\)-bit word \(\mathrm{bin}_q(s^\star)\) into positions
    \(jq,\dots,(j+1)q-1\) of \(y\)
\EndFor
\State \Return \(y\)
\end{algorithmic}
\end{algorithm}

\begin{algorithm}[H]
\caption{\textsc{DecodeBinaryToOneHotAndCheck}}
\label{alg:decode-binary-to-onehot}
\begin{algorithmic}[1]
\Require Binary sample \(y\in\{0,1\}^{nq}\), with \(S=nK\), \(q=\lceil\log_2 S\rceil\);
demands \(d[0{:}n{-}1]\); vehicle capacities \(Q[0{:}K{-}1]\); integers \(n,K\).
\Ensure \textbf{true} iff the decoded sample is feasible in the original one-hot routing model.
\State initialize \(x\gets 0^{nS}\)
\For{\(j\gets 0\) \textbf{to} \(n-1\)}
    \State read block \(y_j := y[jq:(j+1)q-1]\)
    \State \(s_j \gets \mathrm{int}(y_j)\) \Comment{decode binary block to an integer}
    \If{\(s_j \ge S\)}
        \State \Return \textbf{false} \Comment{padding leakage / invalid binary code}
    \EndIf
    \State set \(x[jS + s_j] \gets 1\)
\EndFor
\State \Return \Call{FeasibleGlobalPositions}{$x,d,Q,n,K$}
\end{algorithmic}
\end{algorithm}


\subsection{Hamiltonian reformulation on compressed binary blocks}
\label{subsec:binary-hamiltonian}

The valid binary code space of a single compressed block is
\[
\mathcal K_{\mathrm{val}}
:=
\mathrm{span}\{\ket{\mathrm{bin}_q(s)}: 0\le s\le S-1\}
\subseteq (\mathbb C^2)^{\otimes q}.
\]

The one-hot rank-one symbol projector
\[
X^{(j)}_{i,k}=\ket{(i,k)}\!\bra{(i,k)}
\]
is replaced on the compressed block by the binary computational-basis projector
\begin{equation}
\label{eq:binary-local-projector}
\Pi^{(j)}_{i,k}
:=
\ket{\mathrm{bin}_q(\nu(i,k))}\!\bra{\mathrm{bin}_q(\nu(i,k))}.
\end{equation}
If the binary digits of \(\nu(i,k)\) are
\[
\nu(i,k)=\sum_{r=0}^{q-1} 2^r\,b_r(i,k),
\qquad b_r(i,k)\in\{0,1\},
\]
then \(\Pi^{(j)}_{i,k}\) admits the Pauli-\(Z\) factorization
\begin{equation}
\label{eq:binary-projector-pauli}
\Pi^{(j)}_{i,k}
=
\bigotimes_{r=0}^{q-1}
\frac{I + (-1)^{\,b_r(i,k)} Z^{(j)}_r}{2}.
\end{equation}
Thus every routing term remains diagonal after compression.

\paragraph{Global uniqueness.}
The customer-once penalty becomes
\begin{equation}
\label{eq:binary-once-penalty}
H_{\mathrm{once}}^{\mathrm{bin}}
:=
\lambda_{\mathrm{once}}
\sum_{i=1}^{n}
\Bigl(
\sum_{j=1}^{n}\sum_{k=1}^{K} \Pi^{(j)}_{i,k} - 1
\Bigr)^2.
\end{equation}

\paragraph{Capacity load operators.}
The binary load observable of vehicle \(k\) is
\begin{equation}
\label{eq:binary-load-operator}
\widehat D_k^{\mathrm{bin}}
:=
\sum_{j=1}^{n}\sum_{i=1}^{n} d_i\,\Pi^{(j)}_{i,k},
\end{equation}
so the hinge-square capacity penalty becomes
\begin{equation}
\label{eq:binary-capacity-penalty}
H_{\mathrm{cap}}^{\mathrm{bin}}
:=
\lambda_{\mathrm{cap}}
\sum_{k=1}^{K}
\Bigl(
\bigl[\widehat D_k^{\mathrm{bin}} - Q_k\bigr]_+
\Bigr)^2.
\end{equation}
As in the one-hot formulation, this term is diagonal and ancilla-free.

\paragraph{Routing objective.}
The diagonal travel-cost Hamiltonian on the compressed register is
\begin{align}
\label{eq:binary-routing-objective}
H_{\mathrm{obj}}^{\mathrm{bin}}
:=
\lambda_{\mathrm{obj}}
\Biggl[
&\sum_{j=1}^{n-1}\sum_{i,i'=1}^{n}\sum_{k,k'=1}^{K}
W_{(i,k)\to(i',k')}
\Pi^{(j)}_{i,k}\Pi^{(j+1)}_{i',k'}
\nonumber\\
&\qquad\qquad
+
\sum_{i=1}^{n}\sum_{k=1}^{K}
w_{\mathtt{dep}_k,i}\,\Pi^{(1)}_{i,k}
+
\sum_{i=1}^{n}\sum_{k=1}^{K}
w_{i,\mathtt{dep}_k}\,\Pi^{(n)}_{i,k}
\Biggr].
\end{align}

If \(S\) is not a power of two, the binary block contains invalid padded words
\[
\mathcal V_{\mathrm{pad}}:=\{S,S+1,\dots,2^q-1\}.
\]
To suppress leakage into these words, define
\begin{equation}
\label{eq:padding-projector}
P_{\mathrm{pad}}^{(j)}
:=
\sum_{u=S}^{2^q-1} \ket{u}\!\bra{u}_j,
\qquad
P_{\mathrm{val}}^{(j)}
:=
\sum_{u=0}^{S-1} \ket{u}\!\bra{u}_j,
\end{equation}
and add the diagonal padding term
\begin{equation}
\label{eq:padding-penalty}
H_{\mathrm{pad}}
:=
\lambda_{\mathrm{pad}}
\sum_{j=1}^{n} P_{\mathrm{pad}}^{(j)}.
\end{equation}

The full diagonal binary routing Hamiltonian is therefore
\begin{equation}
\label{eq:binary-total-H}
H_C^{\mathrm{bin}}
=
H_{\mathrm{once}}^{\mathrm{bin}}
+
H_{\mathrm{cap}}^{\mathrm{bin}}
+
H_{\mathrm{obj}}^{\mathrm{bin}}
+
H_{\mathrm{pad}}.
\end{equation}

The compressed Hamiltonian \(H_C^{\mathrm{bin}}\) preserves the diagonal routing semantics of the one-hot construction. Every valid binary computational-basis state encodes one
composite symbol \((i,k)\) per global position, and under the inverse translation introduced in Alg.~\ref{alg:decode-binary-to-onehot} its decoded one-hot representative
is scored by the same routing objective and checked by the same feasibility criteria as in the original formulation.


\subsubsection{Compressed mixers: exact image and practical surrogate}
\label{subsec:binary-mixer}

The problem Hamiltonian compresses exactly because it is diagonal in the symbol basis.
The more delicate step is the reformulation of the block mixer. Let the one-hot block mixer on a local alphabet of size \(S=nK\) be
\begin{equation}
\label{eq:onehot-block-mixer-recall}
\widetilde H_{XY}^{(j)}
=
\frac{1}{S-1}
\sum_{0\le u<v\le S-1}
\Bigl(
\ket{e_u}\!\bra{e_v} + \ket{e_v}\!\bra{e_u}
\Bigr)_j.
\end{equation}
Let
\begin{equation}
\label{eq:block-compression-map}
\mathcal C_{\mathrm{blk}}
:=
\sum_{s=0}^{S-1}
\ket{\mathrm{bin}_q(s)}\!\bra{e_s}
\end{equation}
be the blockwise compression isometry from the one-hot symbol space into the valid binary
code space \(\mathcal K_{\mathrm{val}}\). The exact binary image of the one-hot mixer is then
\begin{equation}
\label{eq:exact-binary-mixer}
\widetilde H_{M,\mathrm{exact}}^{(j)}
:=
\mathcal C_{\mathrm{blk}}\,
\widetilde H_{XY}^{(j)}\,
\mathcal C_{\mathrm{blk}}^\dagger
=
\frac{1}{S-1}
\sum_{0\le u<v\le S-1}
\Bigl(
\ket{u}\!\bra{v} + \ket{v}\!\bra{u}
\Bigr)_j,
\end{equation}
where \(\ket{u}\) and \(\ket{v}\) are understood as \(q\)-qubit binary basis states.
By construction,
\[
\widetilde H_{M,\mathrm{exact}}^{(j)}
=
P_{\mathrm{val}}^{(j)}
\widetilde H_{M,\mathrm{exact}}^{(j)}
P_{\mathrm{val}}^{(j)},
\]
so the valid code space is invariant under the exact compressed mixer. The corresponding exact binary mixer unitary is
\begin{equation}
\label{eq:exact-binary-unitary}
U_{M,\mathrm{exact}}^{\mathrm{bin}}(\beta)
=
\bigotimes_{j=1}^{n}
\exp\!\bigl(-i\beta\,\widetilde H_{M,\mathrm{exact}}^{(j)}\bigr).
\end{equation}
On the valid binary subspace, this dynamics is unitarily equivalent to the original one-hot block-XY mixing.

\paragraph{Compressed CE--QAOA layer.}
With either mixer choice,
\[
U_M^{\mathrm{bin}}(\beta)
\in
\Bigl\{
U_{M,\mathrm{exact}}^{\mathrm{bin}}(\beta),
\;
U_{M,\mathrm{sur}}^{\mathrm{bin}}(\beta)
\Bigr\},
\]
the compressed routing ansatz is
\begin{equation}
\label{eq:binary-ansatz}
\ket{\psi_p^{\mathrm{bin}}(\vec\gamma,\vec\beta)}
=
\Bigl(
\prod_{\ell=1}^{p}
U_M^{\mathrm{bin}}(\beta_\ell)\,
e^{-i\gamma_\ell H_C^{\mathrm{bin}}}
\Bigr)
\ket{s_0^{\mathrm{bin}}},
\end{equation}
where \(\ket{s_0^{\mathrm{bin}}}\) is a chosen initial state supported on the valid binary
code space, for example the uniform valid-code superposition
\[
\ket{s_{\mathrm{blk}}^{\mathrm{bin}}}
=
\frac{1}{\sqrt{S}}
\sum_{s=0}^{S-1}\ket{\mathrm{bin}_q(s)},
\qquad
\ket{s_0^{\mathrm{bin}}}
=
\bigl(\ket{s_{\mathrm{blk}}^{\mathrm{bin}}}\bigr)^{\otimes n}.
\]

After measurement, a binary sample \(y\in\{0,1\}^{nq}\) is translated blockwise back to a
one-hot bitstring \(x\in\{0,1\}^{nS}\) by Alg.~\ref{alg:decode-binary-to-onehot}. The final
accept/reject decision is then delegated unchanged to
\(\textsc{FeasibleGlobalPositions}(x)\).

\subsection{Optimality and performance guarantees}
\label{sec:beyond-classical}

The performance claims of the present routing solver are inherited from the
\emph{CE-QAOA} framework developed in Refs.~\cite{onahce,onahfund,onahfinite}.
Since the CVRP and PDP formulations are realized as specializations of the
Onah--Firt--Michielsen kernel, the same finite-depth, finite-shot, and encoded-manifold guarantees become available under the same instance-dependent conditions. In particular, our global-position construction satisfies the kernel requirements once the capacity-penalization terms are folded into the objective. Thus, the state lives in a block one-hot manifold with $m=n$ blocks and local dimension $n_{\mathrm{loc}}=nK$ for CVRP and $n_{\mathrm{loc}}=TK$ for PDP. The normalized block-XY mixer preserves the encoded manifold with operator norm $O(1)$ on each block. Hence the routing solver inherits the same design-level performance statements proved for CE--QAOA on the encoded space.

The first consequence is the 1-design baseline. Under the block
permutation twirl, one-layer CE--QAOA forms an exact unitary $1$-design on the manifold $\OH$, so that
\[
\mathbb E_U\!\left[\left|\braket{x}{U|\phi}\right|^2\right]=\frac{1}{D},
\qquad D=(nK)^n
\]
for CVRP, and analogously with $D=(TK)^T$ for PDP. At shallow depth, the
interleaving of block-XY evolution with diagonal entanglers yields the
approximate second-moment behavior established in Ref.~\cite{onahce}. This
gives encoded anticoncentration on the native constrained manifold rather than
on the full ambient cube as shown in Fig. \ref{fig:vrp-anticoncentration}.

The second consequence is a direct optimal-solution sampling guarantee in the
Fej\'er-filtered regime. For completeness, we recall the analytic mechanism
underlying the inherited finite-depth and finite-shot success bounds in
Appendix~\ref{app:fejer-mechanism} and
Theorem~\ref{thm:app-fejer-success}. In particular, we summarize the
dephased Fej\'er-filter reference model, define the mixer envelope $W_p$, the
wrapped phase map $\theta$, the optimal-set weight $C_\beta$, and the
off-peak quantity $M_p(\delta)$, and restate the resulting dimension-free
lower bound on the probability of sampling an optimal routing configuration. Whenever the routing
instance has nonzero envelope weight on the optimal set and a positive wrapped
phase separation from competing levels, finite-depth CE--QAOA assigns a
dimension-free lower bound to the probability of sampling an optimum, as
proved in Ref.~\cite{onahfinite}.

\begin{figure}[htbp]
  \centering
  \begin{subfigure}[t]{0.32\linewidth}
    \centering
    \includegraphics[width=\linewidth]{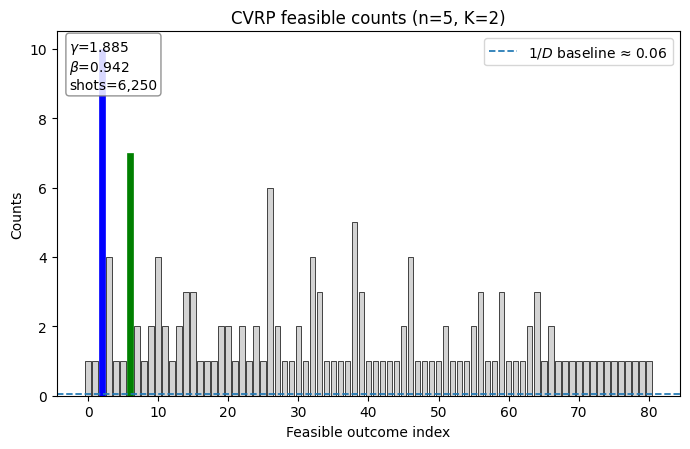}
    \caption{Peak probabilities well above $1/D$. Shots  $\propto n^{3}$.}
  \end{subfigure}\hfill
  \begin{subfigure}[t]{0.32\linewidth}
    \centering
    \includegraphics[width=\linewidth]{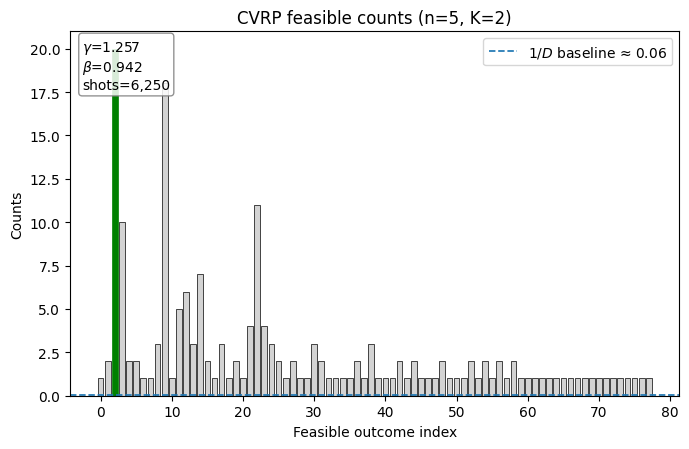}
    \caption{Encoded anticoncentration: several feasible strings $\gg 1/D$.}
  \end{subfigure}\hfill
  \begin{subfigure}[t]{0.32\linewidth}
    \centering
    \includegraphics[width=\linewidth]{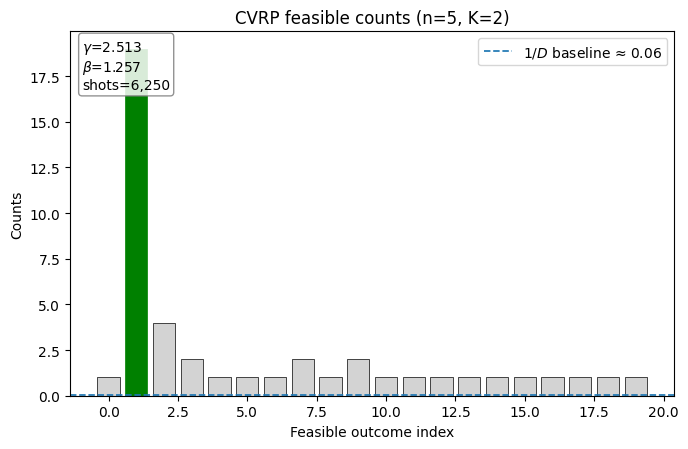}
    \caption{Best-by-cost (green) recovered under $S=\Theta(n^{3})$ sampling.}
  \end{subfigure}

  \caption{\textbf{Encoded anticoncentration and design baseline.}
  Each panel shows the empirical histogram of feasible outcomes from a depth-$p=1$
  CE--QAOA run on a CVRP instance, using shots $S=\Theta(n^{3})$ per grid point. The dashed line is the encoded design baseline $1/D$ with $D=(nK)^{n}$. Peaks far above $1/D$ and the successful identification of the optimum support the anticoncentration guarantees.}
  \label{fig:vrp-anticoncentration}
\end{figure}

\section{A hybrid quantum--classical routing algorithm}
\label{subsec:hybrid-routing}

\subsection{Constraint handling in the hybrid pipeline.}

The native constraints \eqref{eq:block-onehot}--\eqref{eq:item-once}, together with the timeline objective \eqref{eq:cvrp-A-obj}, do not by themselves force the positions assigned to a fixed vehicle \(k\) to form a single contiguous interval on the global timeline. Hence the encoded basis set contains all standard single-route CVRP tours as a subset, but it may also contain segmented same-vehicle timelines unless one adds an explicit contiguity constraint or imposes a decode-time post-selection rule. Consequently, we interpret the global-position formulation as a CE--QAOA-compatible routing model whose standard CVRP sector is selected operationally by feasibility filtering algorithm (Alg. \ref{alg:feasible-global-positions}) introduced next.  Thus, the complete routing pipeline studied in this work is naturally hybrid.  The quantum stage generates candidate bitstrings by sampling the CE--QAOA circuit over a coarse parameter grid, while the classical stage performs exact admissibility checks and exact scoring on the sampled outputs.

For the global-position colored-permutation encoding of CVRP, the key classical routine is a feasibility oracle that filters raw CE--QAOA samples before exact objective evaluation. Given a measured bitstring \(b\in\{0,1\}^{n^2K}\), the oracle checks four conditions: (i) one-hot occupancy in every global-position block, (ii) global uniqueness of each customer, (iii) vehicle-capacity compliance, and (iv) route contiguity on the global timeline. 

\begin{proposition}[Polynomial-time feasibility certification]
\label{prop:feasible-polytime}
Let \(b\in\{0,1\}^{Q_{\mathrm{bits}}}\) with \(Q_{\mathrm{bits}}=n^2K\), let
\(d[0{:}n{-}1]\) be the demand vector, and let \(Q[0{:}K{-}1]\) be the vehicle-capacity vector.
Define
\[
L
\;:=\;
\max\Bigl\{
\max_{0\le i\le n-1}\lceil \log_2(d[i]+1)\rceil,\;
\max_{0\le k\le K-1}\lceil \log_2(Q[k]+1)\rceil
\Bigr\}.
\]
Then Algorithm~\ref{alg:feasible-global-positions} decides feasibility in polynomial time. More precisely:
\begin{enumerate}[label=\textup{(\roman*)},leftmargin=2em]
    \item on a unit-cost RAM (or word-RAM with word size at least \(L+\log n\)), its runtime is
    \[
    O(n^2K);
    \]
    \item on the standard bit model, its runtime is
    \[
    O\!\bigl(n^2K + (n+K)(L+\log n)\bigr).
    \]
\end{enumerate}
\end{proposition}

\begin{proof}
We analyze the steps of Algorithm~\ref{alg:feasible-global-positions}.

\paragraph{Initialization.}
The arrays
\[
\texttt{seen}[0{:}n{-}1],\quad
\texttt{load}[0{:}K{-}1],\quad
\texttt{count}[0{:}K{-}1],\quad
\texttt{firstpos}[0{:}K{-}1],\quad
\texttt{lastpos}[0{:}K{-}1]
\]
are initialized in \(O(n+K)\) time.

\paragraph{Block scan.}
There are exactly \(n\) global-position blocks, and each block has size
\[
S=nK.
\]
For each block \(j\), the inner loop scans all \(S\) entries of that block in the worst case in order to verify one-hotness and locate the unique active symbol. Hence the total number of bit inspections is
\[
n\cdot S = n(nK)=n^2K.
\]
All other operations performed once the active symbol is found---namely decoding
\((i,k)\), checking \(\texttt{seen}[i]\), updating \(\texttt{count}[k]\), and updating
\(\texttt{firstpos}[k]\), \(\texttt{lastpos}[k]\)---take constant time per block on a unit-cost RAM. Thus, apart from arithmetic on the load variables, the full scan costs \(O(n^2K)\).

\paragraph{Capacity accumulation.}
For each of the \(n\) blocks, the algorithm performs one update
\[
\texttt{load}[k] \gets \texttt{load}[k] + d[i].
\]
If the demands and capacities are represented using at most \(L\) bits, then every partial load satisfies
\[
\texttt{load}[k]\le \sum_{i=0}^{n-1} d[i],
\]
so its bit-length is at most \(L+\lceil \log_2 n\rceil\). Therefore, on the bit model, each addition costs
\[
O(L+\log n),
\]
and all \(n\) such additions together cost
\[
O\!\bigl(n(L+\log n)\bigr).
\]

\paragraph{Final checks.}
The final loop over vehicles has length \(K\). For each vehicle \(k\), the algorithm performs:
\begin{enumerate}[label=\textup{(\alph*)},leftmargin=2em]
    \item one capacity comparison \(\texttt{load}[k] > Q[k]\), and
    \item one contiguity check
    \[
    \texttt{lastpos}[k]-\texttt{firstpos}[k]+1 \stackrel{?}{=} \texttt{count}[k].
    \]
\end{enumerate}
The contiguity check uses only indices and counters bounded by \(n\), so it is \(O(1)\) on either model. The capacity comparison costs \(O(1)\) on a unit-cost RAM and
\(O(L+\log n)\) on the bit model. Thus the total cost of the final loop is
\[
O(K)
\quad\text{on a unit-cost RAM,}
\]
and
\[
O\!\bigl(K(L+\log n)\bigr)
\quad\text{on the bit model.}
\]

\paragraph{Total runtime.}
Summing all contributions yields
\[
O(n^2K)
\]
on a unit-cost RAM, and
\[
O\!\bigl(n^2K + n(L+\log n) + K(L+\log n)\bigr)
=
O\!\bigl(n^2K + (n+K)(L+\log n)\bigr)
\]
on the bit model. Proving the claim.
\end{proof}

\begin{corollary}[Polynomial-time classical post-processing in PHQC]
\label{cor:phqc-postproc-poly}
Let \(|\mathcal G|\) be the number of parameter pairs in the coarse CE--QAOA grid and let \(S_{\mathrm{shots}}\) be the number of samples drawn per grid point. Then the total classical feasibility-filtering cost inside PHQC is polynomial. In particular, on a unit-cost RAM it is
\[
O\!\bigl(|\mathcal G|\,S_{\mathrm{shots}}\,n^2K\bigr),
\]
and on the bit model it is
\[
O\!\Bigl(|\mathcal G|\,S_{\mathrm{shots}}
\bigl[n^2K + (n+K)(L+\log n)\bigr]\Bigr).
\]
Hence, whenever \(|\mathcal G|\) and \(S_{\mathrm{shots}}\) are polynomially bounded, the full classical post-processing stage of PHQC is polynomial-time.
\end{corollary}

\begin{proof}
Apply Prop.~\ref{prop:feasible-polytime} independently to each sampled bitstring processed by PHQC. Since PHQC handles \(|\mathcal G|\,S_{\mathrm{shots}}\) sampled bitstrings in total, the result follows by multiplication.
\end{proof}

Algorithm~\ref{alg:feasible-global-positions} is the deterministic admissibility oracle used by the PHQC post-processing layer in App.~\ref{sec:PHQC}. Proposition~\ref{prop:feasible-polytime} shows that this deterministic admissibility test remains polynomial-time, so the feasibility filtering does not alter the polynomial overhead of the hybrid pipeline. For each sampled bitstring \(b\), PHQC first calls \textsc{FeasibleGlobalPositions}\((b)\). If the oracle rejects, the sample is discarded. If the oracle accepts, the classical checker evaluates the exact routing score \(E(b)=\langle b\mid H_{\mathrm{obj}}\mid b\rangle\) (or, if desired, the full diagonal score \(\langle b\mid H_C\mid b\rangle\)) and retains the best feasible sample seen across the entire parameter grid. In this way, the quantum stage is used only to generate candidate solutions, while the classical stage certifies feasibility and performs exact selection. In practice, early exits on a second ``1'' in a block, a repeated customer, or an immediate capacity violation substantially reduce average runtime. The resulting hybrid structure implies that success requires only that an optimal feasible bitstring be sampled at least once, not that it dominates the output distribution.


Beyond its role inside PHQC, Algorithm~\ref{alg:feasible-global-positions} has independent utility as a deterministic polynomial-time decoding and feasibility-certification oracle for structured routing samples. In particular, its use is not tied to CE--QAOA itself. Any quantum or quantum-inspired pipeline that produces candidate bitstrings can use the same oracle to decode samples, certify admissibility, and isolate the standard CVRP sector before exact classical scoring.

\begin{algorithm}[H]
\caption{\textsc{FeasibleGlobalPositions} --- Feasibility check for a sampled bitstring}
\label{alg:feasible-global-positions}
\begin{algorithmic}[1]
\Require Bitstring \(b\in\{0,1\}^{Q_{\mathrm{bits}}}\) with \(Q_{\mathrm{bits}}=n^2K\);
         demands \(d[0{:}n{-}1]\); vehicle capacities \(Q[0{:}K{-}1]\);
         integers \(n,K\).
\Ensure \textbf{true} iff \(b\) satisfies one-hot, uniqueness, capacity, and contiguity.
\State \(S \gets nK\) \Comment{block size (symbols per position)}
\State \textbf{initialize} \(\texttt{seen}[0{:}n{-}1]\gets\) \textbf{false}
\State \textbf{initialize} \(\texttt{load}[0{:}K{-}1]\gets 0\)
\State \textbf{initialize} \(\texttt{count}[0{:}K{-}1]\gets 0\)
\State \textbf{initialize} \(\texttt{firstpos}[0{:}K{-}1]\gets -1\)
\State \textbf{initialize} \(\texttt{lastpos}[0{:}K{-}1]\gets -1\)
\For{\(j \gets 0\) \textbf{to} \(n{-}1\)} \Comment{scan each global position block}
  \State \(L \gets j\cdot S\)
  \State \(\texttt{ones} \gets 0,\; s^\star \gets -1\)
  \For{\(s \gets 0\) \textbf{to} \(S{-}1\)} \Comment{locate the unique active symbol}
    \If{\(b[L+s]=1\)}
      \State \(\texttt{ones} \gets \texttt{ones}+1,\;\; s^\star \gets s\)
      \If{\(\texttt{ones}>1\)}
        \State \Return \textbf{false} \Comment{not one-hot}
      \EndIf
    \EndIf
  \EndFor
  \If{\(\texttt{ones}\neq 1\)}
    \State \Return \textbf{false} \Comment{zero-hot or multi-hot block}
  \EndIf
  \State \(i \gets s^\star \bmod n,\qquad k \gets \lfloor s^\star/n\rfloor\) \Comment{decode customer and vehicle}
  \If{\(\texttt{seen}[i]\)}
    \State \Return \textbf{false} \Comment{customer appears more than once}
  \Else
    \State \(\texttt{seen}[i] \gets \textbf{true}\)
  \EndIf
  \State \(\texttt{load}[k] \gets \texttt{load}[k] + d[i]\)
  \State \(\texttt{count}[k] \gets \texttt{count}[k] + 1\)
  \If{\(\texttt{firstpos}[k] = -1\)}
    \State \(\texttt{firstpos}[k] \gets j\)
  \EndIf
  \State \(\texttt{lastpos}[k] \gets j\)
\EndFor
\For{\(k \gets 0\) \textbf{to} \(K{-}1\)}
  \If{\(\texttt{load}[k] > Q[k]\)}
    \State \Return \textbf{false} \Comment{capacity violation for vehicle \(k\)}
  \EndIf
  \If{\(\texttt{count}[k] > 0\)}
    \If{\(\texttt{lastpos}[k]-\texttt{firstpos}[k]+1 \neq \texttt{count}[k]\)}
      \State \Return \textbf{false} \Comment{vehicle \(k\) appears in multiple disjoint segments}
    \EndIf
  \EndIf
\EndFor
\State \Return \textbf{true}
\end{algorithmic}
\end{algorithm}

\subsection{Polynomial-time Hybrid Quantum--Classical Solver (PHQC)}
\label{sec:PHQC}
We operationalize the full hybrid pipeline with classical post-processing in the routing version of polynomial-time hybrid quantum--classical (PHQC) framework introduced in Ref. \cite{onahce} ( App.~\ref{sec:PHQC}) where the quantum samples generated by CE-QAOA are handed over to Alg. \ref{alg:feasible-global-positions} for feasibility checks and scoring.  The hybrid separation is important algorithmically. The globally optimal feasible solution need not be the most frequent sample. It is enough that the quantum stage assigns it nonzero probability, since a single optimal feasible hit is sufficient for the deterministic checker to recover it. 

\begin{algorithm}[H]
\caption{\textbf{PHQC} --- CE--QAOA grid search with feasibility filtering and scoring}
\label{alg:PHQC}
\begin{algorithmic}[1]
\Require depth \(p\); coarse grid \(\mathcal G\subset[0,\pi]^2\); cost Hamiltonian \(H_C\); objective Hamiltonian \(H_{\mathrm{obj}}\); mixer \(U_M\); initial state \(\ket{s_0}\); shots \(S\) per grid point.
\Ensure best feasible sampled bitstring \(b^\star\) and its score \(E^\star\).
\State \(b^\star \gets \texttt{null}\), \(\;E^\star \gets +\infty\)
\For{each \((\beta,\gamma)\in\mathcal G\)}
    \For{\(r=1,\dots,S\)}
        \State prepare \(\ket{s_0}\)
        \State apply the CE--QAOA circuit for \((\beta,\gamma)\)
        \State measure a bitstring \(b\)
        \If{\textsc{FeasibleGlobalPositions}\((b)\)}
            \State compute \(E(b)=\langle b\mid H_{\mathrm{obj}}\mid b\rangle\)
            \If{\(E(b)<E^\star\)}
                \State \(b^\star \gets b\), \(\;E^\star \gets E(b)\)
            \EndIf
        \EndIf
    \EndFor
\EndFor
\State \Return \((b^\star,E^\star)\)
\end{algorithmic}
\end{algorithm}

We further extend the pipeline to include a coarse parameter-grid search. A single appearance is sufficient for exact recovery.

Let \(p\) be the circuit depth and let
\[
\mathcal G
=
\{(\beta_a,\gamma_b): 0\le a\le N_\beta,\ 0\le b\le N_\gamma\}
\subset [0,\pi]\times[0,\pi]
\]
be a rectangular coarse grid, for example
\[
\beta_a = a\,\Delta_\beta,\qquad
\gamma_b = b\,\Delta_\gamma,
\qquad
\Delta_\beta = \frac{\pi}{N_\beta},\quad
\Delta_\gamma = \frac{\pi}{N_\gamma}.
\]
For each \((\beta,\gamma)\in\mathcal G\), PHQC prepares the CE--QAOA state
\[
\ket{\psi_p(\beta,\gamma)}
=
U_M(\beta)\,e^{-i\gamma H_C}\ket{s_0}
\]
(for \(p=1\); or depth-\(p\) generalization), samples it \(S\) times, and feeds each measured bitstring into the deterministic feasibility oracle
\(\textsc{FeasibleGlobalPositions}\) from Alg.~\ref{alg:feasible-global-positions}.

\subsection{Numerical Implementation}
\label{sec:impl}

On each block (local size $nK$) we use the normalized XY mixer restricted to the one–hot sector $\mathcal H_1^{(nK)}$. The full mixer layer is the block tensor product
\begin{equation}
\label{eq:xy-full}
U_M(\beta)\;=\;\bigotimes_{b=1}^{n} U_M^{(b)}(\beta).
\end{equation}
The per–block initial state is the uniform one–hot (a $W$–state) over $nK$ symbols,
\begin{equation}
\label{eq:start-block}
\ket{s_{\mathrm{blk}}}
\;=\;
\frac{1}{\sqrt{nK}}\sum_{(i,k)\in[n]\times[K]}\ket{e_{(i,k)}},
\qquad
\ket{s_0}\;=\;\ket{s_{\mathrm{blk}}}^{\otimes n}.
\end{equation}

A depth–$p$ CE--QAOA circuit alternates diagonal cost and the block–XY mixer:
\begin{equation}
\label{eq:ansatz}
\ket{\psi_p(\vec\gamma,\vec\beta)}
\;=\;
\Bigl(\prod_{\ell=1}^{p} U_M(\beta_\ell)\,e^{-i\gamma_\ell H_C}\Bigr)\ket{s_0}.
\end{equation}
In the implementation we instantiate $H_C$ from a QUBO dictionary via the
Ising map $x\mapsto(1{-}Z)/2$, yielding a $Z/ZZ$ Pauli operator.
We use $p=1$ and grid–search $(\gamma,\beta)$ on a coarse lattice
($O(nK)$ points per axis); the shot budget follows the $S=50\,(nK)^3$ rule. The circuit primitives with a qiskit wrapper \cite{Qiskit2023}  were first reported in \cite{onahce}. As in Ref. \cite{onahce}, we used the \texttt{matrix\_product\_state} backend of the IBM Qiskit \textsc{Aer} simulator\cite{Qiskit2023}. For reproducibility, we configured the backend with an \emph{aggressive truncation} regime with maximum bond dimension of $128$, truncation and validation thresholds of $10^{-3}$, and an amplitude--chopping threshold of $10^{-3}$. These circuit simulations are in the sense of Vidal’s tensor–network formulation~\cite{Vidal2003Efficient}. We collect implementation details in App. \ref{app:impl}.

\paragraph{Goal and evaluation protocol.}
Our numerical study is designed to evaluate the end-to-end performance of the proposed hybrid routing pipeline in Alg. \ref{alg:PHQC} on small benchmark CVRP instances while preserving the modeling choices of this paper. For each instance, we run depth $p=1$ CE--QAOA over a fixed $(\gamma,\beta)$ grid, sample the resulting circuit at each grid point, post-select feasible solutions and report the best feasible routing cost observed. We benchmark on the QOPTLib family instances used in related hybrid QUBO studies\cite{Osaba2024Qoptlib}.

Each instance contains
(i) the number of non-depot customers $n$, (ii) customer coordinates,
(iii) depot coordinate, (iv) demand vector $(d_i)_{i=1}^n$, and (v) per-vehicle
capacity $Q$, with a fixed number of vehicles $K=2$ for all instances. Distances are computed from the Euclidean
metric using precomputed matrices:
\[
W_{i,i'}=\|c_i-c_{i'}\|_2,\qquad
w_{\mathsf D\to i}=\|c_{\mathsf D}-c_i\|_2,\qquad
w_{i\to\mathsf D}=\|c_i-c_{\mathsf D}\|_2,
\]
so that subsequent cost evaluation is $O(1)$ per queried edge. See  additional implementation details in App. \ref{app:impl}. Results of the numerical study are collected in Table \ref{tab:numerical}.  Results of the numerical study are collected in
Tables~\ref{tab:numerical} and \ref{tab:numerical-binary}.  Table~ \ref{tab:numerical} reports the one-hot colored-permutation realization, while
Table~\ref{tab:numerical-binary} reports the corresponding
binary-compressed realization.

The QOptLib values shown are taken from the \emph{Hybrid} column in their report, where the
underlying QUBO is decomposed into smaller subproblems and solved via a hybrid quantum--classical pipeline. In contrast, we treat each instance as a single, non-decomposed problem within our  global-position formulation. The best feasible cost returned by our PHQC pipeline matched the independently verified optimum for every instance in Table~\ref{tab:numerical}. Optimality was verified separately using a branch-and-bound solver implemented in Gurobi \cite{gurobi}. The \(\dagger\) symbol in the table indicates that the value reported in Ref.~\cite{Osaba2024Qoptlib} lies below the optimum verified by our classical check. Both the quantum and classical implementations are made publicly available in \cite{onahcvrpdata}.

\begin{table}[t]
\centering
\caption{Performance on CVRP instances ($K=2$) proposed in Ref.~\cite{Osaba2024Qoptlib}. The qubit count is $Kn^2$ where $n$ is the number of cities and $K=2$ is the number of vehicles. The values in the ``This work'' column are the best feasible costs returned
by PHQC and coincide with the optima independently verified by a Gurobi
branch-and-bound check. The \(\dagger\) symbol marks a QOPTLib Hybrid value
that lies below the optimum certified under our feasibility and scoring
convention. }
\label{tab:numerical}
\begin{tabular}{lccc}
\toprule
Instance & Qubits & QOPTLib (Hybrid) & This work \\
\midrule
P-n41.vrp & 32  & 97  & 69 \\
P-n42.vrp & 32  & 121 & 96 \\
P-n51.vrp & 50  & 94  & 94 \\
P-n52.vrp & 50  & 295 & 295 \\
P-n61.vrp & 72  & 118 & 118 \\
P-n62.vrp & 72  & 122 & 121 \\
P-n71.vrp & 98  & $119^{\dagger}$ & 132 \\
P-n72.vrp & 98  & 164 & 163 \\
P-n81.vrp & 128 & 153 & 136 \\
P-n82.vrp & 128 & 269 & 225 \\
\bottomrule
\end{tabular}
\end{table}

\begin{table}[t]
\centering
\caption{Performance on CVRP instances ($K=2$) proposed in Ref.~\cite{Osaba2024Qoptlib} under the implementation of binary-compressed colored-permutation encoding proposed in sec \ref{subsec:binary-compression}. The qubit count is $n\lceil \log_2(Kn)\rceil$, where $n$ is the number of cities and $K=2$ is the number of vehicles.}
\label{tab:numerical-binary}
\begin{tabular}{lccc}
\toprule
Instance & Qubits & QOPTLib (Hybrid) & This work \\
\midrule
P-n41.vrp & 12 & 97  & 69  \\
P-n42.vrp & 12 & 121 & 96  \\
P-n51.vrp & 20 & 94  & 94  \\
P-n52.vrp & 20 & 295 & 295 \\
P-n61.vrp & 24 & 118 & 118 \\
P-n62.vrp & 24 & 122 & 121 \\
P-n71.vrp & 28 & $119^{\dagger}$ & 132 \\
P-n72.vrp & 28 & 164 & 163 \\
P-n81.vrp & 32 & 153 & 136 \\
P-n82.vrp & 32 & 269 & 225 \\
\bottomrule
\end{tabular}
\end{table}


\subsection{Qubit efficiency compared to prior CVRP formulations.}
\paragraph{Counting sanity check.}
There are $n!$ permutation matrices $P$ over $n$ locations. For each $P$, we may choose any of $K$ trucks
independently for each of the $n$ positions, so there are $K^n$ colorings.
Thus the number of feasible bitstrings is $n!\,K^n$.
(These are the computational-basis states satisfying \eqref{eq:block-onehot}--\eqref{eq:item-once}.) On each block (position $j$), the local Hilbert space is $\mathcal H_1\cong\mathbb C^{nK}$
spanned by the basis $\{\ket{(i,k)}\}$ with \emph{one} excitation
(one-hot over $(i,k)$).\newline The full encoded space is the tensor product
$\OH=\mathcal H_1^{\otimes n}$, whose computational basis states are the
bitstrings satisfying the per-block one-hot condition \eqref{eq:block-onehot}. The initial state
\begin{equation}
\ket{s_0} \;=\; \Bigl(\tfrac{1}{\sqrt{nK}}\sum_{(i,k)} \ket{(i,k)}\Bigr)^{\otimes n}
\end{equation}
is the equal-amplitude superposition over \emph{all} such block-one-hot bitstrings. Adding the diagonal assignment penalty for \eqref{eq:item-once} and the (diagonal) travel cost does not change the computational basis; it only rephases amplitudes. Therefore, the encoded manifold \emph{contains} (as basis states) \emph{all} “colored permutations” $X$ characterized in Theorem~\ref{thm:colored-perms}, and the CE--QAOA mixer  preserves this manifold while the diagonal cost couples blocks. 

In contrast to several quantum-assisted CVRP formulations that \emph{explicitly add auxiliary binary variables} for capacity handling, our formulation is qubit efficient. In the hybrid, cluster-first approach of Feld \emph{et al.}, the capacity constraint in the knapsack-style clustering phase is modeled with extra selector/slack variables, thereby increasing the number of binary decision variables (and hence logical qubits) used by the QUBO formulation \cite{FeldEtAl2018CVRP}. Similarly, Palackal \emph{et al.} introduce additional binary variables to encode vehicle-load levels within their quantum-assisted pipeline for CVRP, which again raises the logical-qubit count required by the constraint representation \cite{PalackalEtAl2023QCVRP,bentley2022q}. In both cases, the qubit overhead scales with the granularity of the capacity encoding (e.g., unary/one-hot versus binary-coded load). By contrast, our protocol maintains a logical-qubit footprint fixed by the problem’s decision variables with visible scaling benefits in the near term non-asymptotic regimes (See Fig. \ref{fig:routing-scale}). See further discussion on industry sized quantum routing in Sec. \ref{sec:indu}.

\begin{figure}[t]
    \centering
    \includegraphics[width=\linewidth]{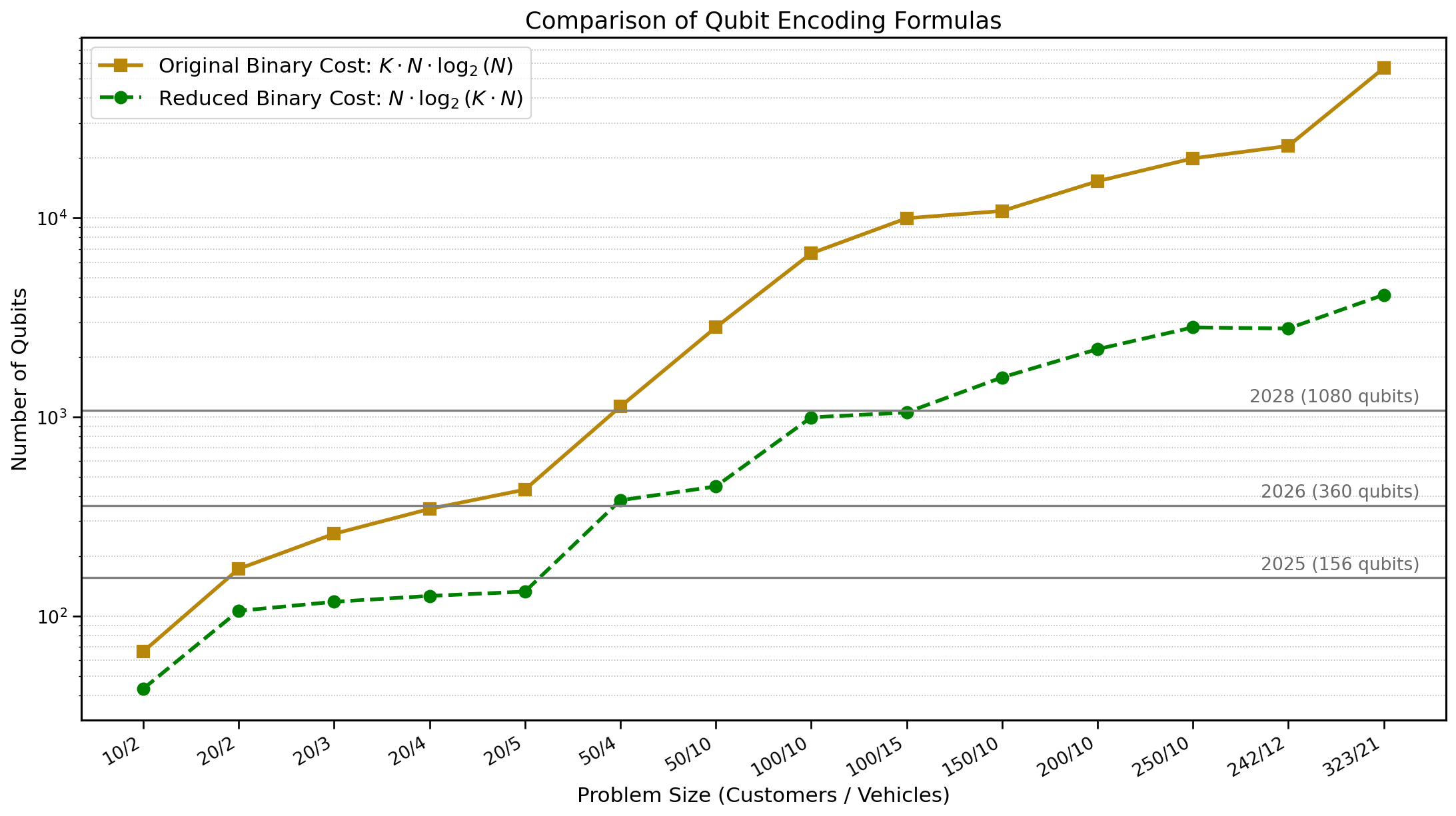}
    \caption{Comparison of qubit counts for two binary global-position encodings across representative CVRP instance sizes. The original separated encoding scales as \(Q_{\mathrm{sep}} = K N \lceil \log_2 N \rceil\), whereas the reduced colored-permutation encoding scales as \(Q_{\mathrm{red}} = N \lceil \log_2(KN) \rceil\). The horizontal lines indicate illustrative hardware thresholds. The reduced encoding changes the linear factor of \(K\) in the number of qubits to a logarithmic dependence and shifts several small-scale industrial routing regimes into the few-hundred to \(\sim 10^3\)-qubit range.}
    \label{fig:routing-scale}
\end{figure}

\subsection{Discussion:  near term quantum routing optimization at the industrial scale}
\label{sec:indu}

The main obstacle to scaling the present routing formulation is the qubit cost of the local alphabet representation. In the one-hot global-position encoding used in this work, each of the \(n\) position blocks carries an alphabet of size \(S=nK\), so the total qubit count is
\[
Q_{\text{1hot}} \;=\; nS \;=\; n^2K.
\]
This representation is theoretically attractive because it is aligned with the Onah--Firt--Michielsen (OFM) kernel with well studied performance guarantees\cite{onahce, onahfund,onahfinite}. However, the quadratic-in-\(n\) footprint makes one-hot the dominant bottleneck long before circuit depth or shot count becomes the limiting factor. To push quantum routing toward industrial scale instances in the near term, the most direct route is to preserve the \emph{global-position} viewpoint while compressing the \emph{local alphabet} from one-hot to binary. A naive binary implementation with vehicle-separated position registers scales as
\[
Q_{\text{sep}} \;\approx\; Kn\log_2 n,
\]
whereas the colored-permutation viewpoint reduces this to
\[
Q_{\text{col}} \;\approx\; n\log_2(nK),
\]
up to the usual ceiling effects. The reduction comes from encoding a single composite symbol \((i,k)\in[n]\times[K]\) in each position block, instead of carrying a separate position register for each vehicle. In other words, the colored-permutation representation removes an unnecessary factor of \(K\) from the number of blocks while keeping the same global-ordering logic that makes the cost and capacity structure diagonal and ancilla-free. 

In the present work we also implement the binary-compressed colored-permutation representation operationally and decode measured binary samples back into the original one-hot routing tensor before certification. In this way, the same deterministic feasibility oracle used throughout the paper remains the final admissibility backend for both realizations. The one-hot formulation therefore supplies the native Onah--Firt--Michielsen (OFM) kernel realization and its inherited guarantee structure, while the binary formulation supplies the qubit-compressed execution layer. On the tested QOPTLib benchmark instances, both realizations recover the independently verified optima under the same final certification rule. This closes the loop between the theorem-aligned one-hot construction and the practically compressed routing representation.

This compression is already substantial in the regime most relevant for early industrial use cases. As illustrated in Fig.~\ref{fig:routing-scale}, instances with \(50\) to \(100\) customer locations and \(5\) to \(10\) vehicles move from the multi-thousand-qubit regime under the separated binary encoding into the few-hundred to roughly \(10^3\)-qubit regime under the reduced colored-permutation binary encoding. For example, at \((n,K)=(50,10)\), the separated encoding requires about \(4\times 10^3\) qubits, while the reduced encoding requires only about \(4.5\times 10^2\). These small-scale but operationally meaningful routing instances enter a logical-register regime where near-term hardware experiments become concrete, provided the compressed mixer and diagonal phase-separator can be implemented with sufficient fidelity.


\paragraph{Regimes of quantum utility.}

A natural way to calibrate near-term quantum utility in routing is to ask whether a quantum-enabled workflow can close open optimality gaps on benchmark instances that remain nontrivial at the current classical frontier. Classical exact methods for CVRP are already remarkably strong. Modern branch-cut-and-price solvers can solve many instances with up to a few hundred customers and have pushed the proven-optimal frontier far beyond the historical \(150\)-customer regime \cite{Pessoa2020,bentley2022q}. Likewise, modern heuristic frameworks routinely achieve very small gaps with modest runtimes on medium-scale routing benchmarks, often remaining below ten minutes on average for problems in the \(100\)--\(200\) customer range \cite{Vidal2014}. Against that backdrop, the ORTEC benchmark instances \texttt{ORTEC-n242-k12} and \texttt{ORTEC-n323-k21}---the ``242/12'' and ``323/21'' cases highlighted in Fig.~\ref{fig:routing-scale}---are especially interesting. They are real-world CVRP instances from the DIMACS/CVRPLIB collection, derived from a US-based grocery delivery service, and their incumbent objective values are listed in CVRPLIB as best known rather than proven optimal \cite{DIMACS2021,CVRPLIBInstances}. In a broader benchmark sense, several instances from the classical Golden set, including \texttt{Golden\_1}, \texttt{Golden\_5}, and \texttt{Golden\_6}, also remain without an optimality certificate at their current incumbent values \cite{Golden1998}. A quantum routing pipeline that either certifies optimality at the current incumbent value or produces a strictly better verified feasible solution for one of these unresolved instances could reasonably mark the onset of \emph{quantum utility} by contributing a verifiable improvement on a benchmark that remains open at the current classical record for decades.

These two realizations separate the roles of structure and scale in a useful way. The one-hot construction provides the clean encoded-manifold formulation in which the CE--QAOA reduction is exact and the routing semantics are transparent. The binary construction preserves the same global-position routing logic at substantially reduced qubit cost and feeds into the same exact post-processing layer after inverse translation. Consequently, the remaining obstacle to larger routing instances is no longer the logical representation itself, but the ability of near-term hardware to execute the compressed circuit with sufficient fidelity. This is precisely the regime in which the present formulation makes the question of quantum routing utility concrete and testable.

\section{Conclusion}

We introduced a global-position colored-permutation formulation for capacitated vehicle routing in which \(n\) positions choose symbols \((i,k)\in[n]\times[K]\) such that each customer appears exactly once, and vehicle labels partition the resulting permutation into \(K\) disjoint partial permutations. This viewpoint yields a compact and operational decoding where measured bitstrings reshape into an \(n\times K\times n\) tensor whose \(K\) slices are per-vehicle partial permutation matrices, from which ordered route segments are read off on the global timeline. On the algorithmic side, we instantiated the Onah--Firt--Michielsen (OFM) kernel as the assignment-and-ordering formulation of CVRP. The initial state spans all block-one-hot assignments; the normalized block-XY mixer preserves the block-one-hot manifold and mixes within it; and the diagonal routing Hamiltonian couples blocks through both assignment penalties and routing edges, ensuring entangling phase separation. A central practical outcome is that capacity enforcement is naturally ancilla-free in the logical encoding. For heterogeneous-demand instances, however, capacity is best viewed as an additional diagonal routing constraint, enforced either by post-selection or by instance-dependent diagonal biasing, rather than as part of the fully label-symmetric penalty sector. 

A brief account of the analytic mechanism behind the performance guarantees is given in Appendix~\ref{app:fejer-mechanism}. In particular, we recall how the dephased Fej\'er-filter reference model isolates the phase-selection mechanism responsible for the dimension-free success bound and clarifies the instance-dependent quantities controlling finite-depth and finite-shot optimality in the CVRP formulation. Finally, we highlighted a clear scaling path obtained by replacing one-hot blocks with binary-coded symbol registers, thereby reducing the qubit footprint from \(n^2K\) to \(\Theta(n\log(nK))\) while retaining the global-position colored-permutation structure. Achieving this at full strength will require new symbol-respecting mixers on binary registers that preserve the favorable CE--QAOA dynamics on the encoded manifold. Developing and benchmarking such mixers, together with sharper conditions under which contiguous single-route solutions are selected from the global-position formulation, appears to be the most direct route to pushing quantum routing optimization to substantially larger instance sizes without sacrificing the kernel-level advantages that motivate the approach.

\section*{Data Availability} All data in this paper and the Python implementation in Qiskit are made available here \url{https://doi.org/10.5281/zenodo.18798145}.

\section*{Conflict of Interests}
All authors declare no competing interests.

\section*{Author contributions}
C.O. wrote the main paper and prepared all the figures, K.M. provided guidance during the project. All the authors reviewed the paper.

\section*{Funding}
This research received no external funding.

\section*{Additional information}
All our results are reported in the main paper.

\textbf{Correspondence and requests for materials} should be addressed to C.O.\ (\texttt{chinonso.calistus.onah@volkswagen.de}).

\appendix

\section{Pickup--Delivery Problem (PDP)}
\label{subsec:pdp-global}

\paragraph{Sets and data.}
There are \emph{atomic tours} \(t\in[T]\); each tour \(t\) internally represents a pickup location \(p_t\) and a delivery location \(d_t\) (pickup and delivery happen within the same job). Trucks \(k\in[K]\) have depots \(\mathtt{dep}_k\) (potentially distinct). Define the inter-tour cost
\[
  \widetilde w_{t\to t'}\ :=\ \mathrm{dist}\bigl(d_t,\ p_{t'}\bigr),
\]
possibly asymmetric.
Depot connections are
\[
w_{\mathtt{dep}_k\to t}:=\mathrm{dist}(\mathtt{dep}_k,p_t),
\qquad
w_{t\to \mathtt{dep}_k}:=\mathrm{dist}(d_t,\mathtt{dep}_k).
\]
When capacities are not modeled, one may set the resource weights \(d_t\) to zero in the constraints so that the objective alone drives the optimization.

\paragraph{Blocks, symbols, one-hot manifold.}
As in CVRP, take \(m:=T\) global positions \(j\in[T]\) and the local alphabet
\[
\mathcal A_{\mathrm{PDP}} \;=\; \{(t,k)\,:\, t\in[T],\,k\in[K]\},
\qquad
n_{\mathrm{loc}}^{\mathrm{PDP}}=TK.
\]
Let \(X^{(j)}_{t,k}\) be the projector for symbol \((t,k)\) on block \(j\).

\paragraph{Global uniqueness and block one-hot.}
Each atomic tour appears exactly once globally, so we impose
\begin{align}
\label{eq:pdp-once}
H_{\mathrm{once}}^{\textsc{B}}
&:=\ \lambda_{\mathrm{once}} \sum_{t=1}^{T}
\Bigl(\ \sum_{j=1}^{T}\sum_{k=1}^{K} X^{(j)}_{t,k} \;-\; 1\ \Bigr)^{2}.
\end{align}
The per-block one-hot constraint
\begin{align}
\label{eq:pdp-block}
H_{\mathrm{blk}}^{\textsc{B}}
&:=\ \lambda_{\mathrm{blk}} \sum_{j=1}^{T}
\Bigl(\ \sum_{t=1}^{T}\sum_{k=1}^{K} X^{(j)}_{t,k} \;-\; 1\ \Bigr)^{2}
\end{align}
is redundant in the hard-coded one-hot encoding and may therefore be omitted, i.e.\ one may set \(H_{\mathrm{blk}}^{\textsc{B}}=0\).

\paragraph{Objective (dead-miles between tours + depot endpoints).}
Define the PDP adjacency cost
\[
  \widetilde W_{(t,k)\to(t',k')}
  \;:=\;
  \begin{cases}
    \widetilde w_{t\to t'} & \text{if } k=k',\\[2pt]
    w_{t\to \mathtt{dep}_k}+w_{\mathtt{dep}_{k'}\to t'} & \text{if } k\neq k'.
  \end{cases}
\]
Then
\begin{align}
\label{eq:pdp-obj}
H_{\mathrm{obj}}^{\textsc{B}}
\;:=\;
\lambda_{\mathrm{obj}}\Biggl[
&\sum_{j=1}^{T-1}\sum_{t,t'=1}^{T}\sum_{k,k'=1}^{K}
  \widetilde W_{(t,k)\to(t',k')} \; X^{(j)}_{t,k}\,X^{(j+1)}_{t',k'}\\[-1pt]
&\quad + \sum_{t=1}^{T}\sum_{k=1}^{K} w_{\mathtt{dep}_k\to t}\; X^{(1)}_{t,k}
\quad + \sum_{t=1}^{T}\sum_{k=1}^{K} w_{t\to \mathtt{dep}_k}\; X^{(T)}_{t,k}
\Biggr].\nonumber
\end{align}

If desired, add a per-truck capacity/duty-time hinge-square penalty, exactly analogous to \eqref{eq:cvrp-A-cap}:
\[
H_{\mathrm{cap}}^{\textsc{B}}
\;:=\;
\lambda_{\mathrm{cap}}\sum_{k=1}^{K}
\Biggl(
\Bigl[
\sum_{j=1}^{T}\sum_{t=1}^{T} d_t\,X^{(j)}_{t,k} \;-\; Q_k
\Bigr]_+
\Biggr)^{2},
\]
or replace \(d_t\) by a duty-time or resource weight.

\begin{equation}
\label{eq:pdp-total}
H_C^{\textsc{B}} \;:=\; H_{\mathrm{once}}^{\textsc{B}} + H_{\mathrm{obj}}^{\textsc{B}} + H_{\mathrm{cap}}^{\textsc{B}}.
\end{equation}
Mixer and initial state are as in CVRP, with \(n_{\mathrm{loc}}\) replaced by \(n_{\mathrm{loc}}^{\mathrm{PDP}}=TK\) and \((i,k)\) replaced by \((t,k)\). The CE--QAOA stack is
\[
\ket{\psi_p(\vec\gamma,\vec\beta)}
=\bigl(\prod_{\ell=1}^p U_M(\beta_\ell)\,e^{-i\gamma_\ell H_C^{\textsc{B}}}\bigr)\ket{s_0}.
\]
Cross-block couplings in \(H_C^{\textsc{B}}\) again make the diagonal layer entangling.

\paragraph{Decoding.}
From a measured bitstring, read for each \(j\) the unique \((t,k)\) with value \(1\).
For truck \(k\), the subsequence of such positions in increasing \(j\) lists its assigned tours in order. Depot edges are implicitly added at boundaries via \eqref{eq:pdp-obj}. No explicit precedence is needed because each \(t\) already encapsulates pickup and delivery.

\subsection{Primer: Permutation matrices and bipartite view}
\label{sec:color}
Let $\sigma:\{1,\dots,n\}\to\{1,\dots,n\}$ be a permutation. Its permutation matrix $P\in\{0,1\}^{n\times n}$ is defined by
\[
  P_{i,j}=\begin{cases}
    1, & j=\sigma(i),\\
    0, & \text{otherwise}.
  \end{cases}
\]
Equivalently, interpret rows $L=\{1,\dots,n\}$ and columns $R=\{1,\dots,n\}$ as the left and right vertex sets of $K_{n,n}$. Then the support of $P$,
\(
  \supp(P)=\{(i,\sigma(i)):\ i\in L\}\subseteq L\times R,
\)
is a perfect matching.

\noindent
The cycle decomposition of $\sigma$ corresponds to directed cycles in the bipartite incidence picture via alternation of $L$- and $R$-indices.

\begin{definition}[Colored decomposition]
A \emph{colored decomposition} of a permutation matrix \(P\in\{0,1\}^{n\times n}\) is a collection of \((0,1)\)-matrices \(\{P^{(c)}\}_{c=1}^{C}\) such that:
\begin{enumerate}[label=(\roman*),leftmargin=2em]
  \item \emph{Partition:} the supports are pairwise disjoint and cover \(\supp(P)\):
  \[
      \supp(P)=\bigsqcup_{c=1}^C \supp\bigl(P^{(c)}\bigr),
      \qquad
      P=\sum_{c=1}^C P^{(c)}.
  \]
  \item \emph{Partial permutations:} each \(P^{(c)}\) has at most one \(1\) in every row and every column; equivalently, each \(P^{(c)}\) is a partial permutation matrix.
\end{enumerate}
We visualize the class index \(c\) as a color. Each color class encodes a partial matching on the support of \(P\).
\end{definition}

\begin{remark}[Cycle-based canonical coloring]
\label{rem:cycle_coloring}
The cycle decomposition of a permutation \(\sigma\) induces a distinguished \emph{special case} of a colored decomposition by assigning one color to each disjoint cycle. Concretely, if
\[
  \sigma=\mathcal{C}_1\,\mathcal{C}_2\cdots \mathcal{C}_C
\]
is a product of disjoint cycles, then
\[
P=\sum_{c=1}^C P^{(c)},
\]
where \(P^{(c)}\) is the permutation matrix supported on the indices belonging to the cycle \(\mathcal{C}_c\), and zero elsewhere. In this special cycle-coloring case, each nonzero color block is a permutation on its own support and decomposes into circulant cycle blocks.
\end{remark}

\begin{proposition}[Basic properties]
Let \(P=\sum_{c=1}^C P^{(c)}\) be a colored decomposition.
\begin{enumerate}[label=(\alph*),leftmargin=2em]
  \item \emph{Orthogonality of supports:} for \(c\neq d\), one has
  \[
  P^{(c)}\odot P^{(d)}=0.
  \]
  \item \emph{Row/column sparsity:} for each \(c\), every row and every column contains at most one \(1\) in \(P^{(c)}\).
  \item \emph{Special cycle-coloring case:} if the decomposition is the canonical cycle-coloring of Remark~\ref{rem:cycle_coloring}, then each color class is block diagonal, with one circulant permutation block per cycle, and the eigenvalues contributed by a cycle of length \(k\) are the \(k\)-th roots of unity.
\end{enumerate}
\end{proposition}

As an example, consider $n=5$ and the permutation $\sigma=(1\ 3\ 4)(2\ 5)$ (written in cycle notation). The permutation matrix $P$ has $1$'s at positions
\[
(1,3),\ (3,4),\ (4,1)\quad\text{and}\quad (2,5),\ (5,2).
\]

\[
P \;=\;
\begin{bmatrix}
0&0&1&0&0\\
0&0&0&0&1\\
0&0&0&1&0\\
1&0&0&0&0\\
0&1&0&0&0
\end{bmatrix}
\]

A two-color decomposition $P=P^{(1)}+P^{(2)}$ is:
\[
  P^{(1)}=\begin{bmatrix}
  0&0&1&0&0\\
  0&0&0&0&0\\
  0&0&0&1&0\\
  1&0&0&0&0\\
  0&0&0&0&0
  \end{bmatrix},
  \qquad
  P^{(2)}=\begin{bmatrix}
  0&0&0&0&0\\
  0&0&0&0&1\\
  0&0&0&0&0\\
  0&0&0&0&0\\
  0&1&0&0&0
  \end{bmatrix},
  \qquad
  P=P^{(1)}+P^{(2)}.
\]
Here $P^{(1)}$ (red color) encodes the 3-cycle $(1\ 3\ 4)$ and $P^{(2)}$ (blue color) encodes the 2-cycle $(2\ 5)$.

\begin{center}
\begin{tikzpicture}[x=1.7cm,y=0.9cm,>=Latex]
  \foreach \i in {1,...,5} {
    \node (L\i) at (0,6-\i) {$L_{\i}$};
    \node (R\i) at (3,6-\i) {$R_{\i}$};
  }
  \draw[very thick,red]   (L1) -- (R3);
  \draw[very thick,red]   (L3) -- (R4);
  \draw[very thick,red]   (L4) -- (R1);
  \draw[very thick,blue]  (L2) -- (R5);
  \draw[very thick,blue]  (L5) -- (R2);
  \node[red]  at (0.1,-0.2) {color 1 (3-cycle)};
  \node[blue] at (3.4,-0.2) {color 2 (2-cycle)};
\end{tikzpicture}
\end{center}

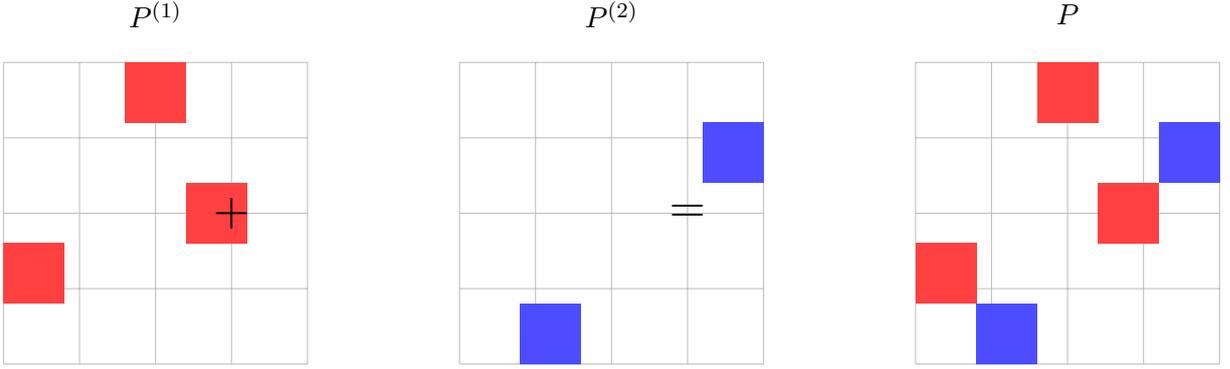
\begin{figure}[H]\centering
\begin{tikzpicture}[x=0.8cm,y=0.8cm]
  \def\n{5}
  \coordinate (L) at (0,0);
  \coordinate (M) at (7.5,0);
  \coordinate (R) at (15.0,0);

  \begin{scope}[shift={(L)}]
    \node at (\n/2,\n+0.8) {$P^{(1)}$};
    \draw[gray!50] (0,0) grid (\n,\n);
    \fill[red!75] (2,4) rectangle ++(1,1);
    \fill[red!75] (3,2) rectangle ++(1,1);
    \fill[red!75] (0,1) rectangle ++(1,1);
  \end{scope}

  \begin{scope}[shift={(M)}]
    \node at (\n/2,\n+0.8) {$P^{(2)}$};
    \draw[gray!50] (0,0) grid (\n,\n);
    \fill[blue!70] (4,3) rectangle ++(1,1);
    \fill[blue!70] (1,0) rectangle ++(1,1);
  \end{scope}

  \begin{scope}[shift={(R)}]
    \node at (\n/2,\n+0.8) {$P$};
    \draw[gray!50] (0,0) grid (\n,\n);
    \fill[red!75]  (2,4) rectangle ++(1,1);
    \fill[red!75]  (3,2) rectangle ++(1,1);
    \fill[red!75]  (0,1) rectangle ++(1,1);
    \fill[blue!70] (4,3) rectangle ++(1,1);
    \fill[blue!70] (1,0) rectangle ++(1,1);
  \end{scope}

  \node[scale=1.6] at ($ (L)!.5!(M) + (0,\n/2) $) {$+$};
  \node[scale=1.6] at ($ (M)!.5!(R) + (0,\n/2) $) {$=$};
\end{tikzpicture}
\caption{Colored decomposition: \(P=P^{(1)}+P^{(2)}\) with disjoint supports.}
\label{fig:colored-decomposition}
\end{figure}

\section{CVRP Examples: from locations to colored permutations}
\label{subsec:worked-mini}

\subsection{Example A: 4 locations 2 Vehicles}
Consider a routing problem with \(\boldsymbol{4}\) locations, \(\boldsymbol{2}\) vehicles. \(\Rightarrow n=3\) customers with  \(K=2\) vehicles. We have \(n=3\) blocks and local alphabet size \(nK=3\cdot2=6\). Total binary variables/qubits \(n^2K=18\).
Use the same distance construction as above. Assume unit demand and take \(d_i\equiv 1\) set  \(\texttt{CAPACITY}=3\), then any single truck can carry up to 3 customers.

\smallskip
\emph{Variables.} \(x_{i,j,k}\) with \(i,j\in\{1,2,3\}\), \(k\in\{1,2\}\).

\emph{Constraints.}
\begin{align}
\label{eq:const}
\text{(block one–hot)}\;\;
&\sum_{i=1}^3\sum_{k=1}^2 x_{i,1,k}=1,\quad
  \sum_{i=1}^3\sum_{k=1}^2 x_{i,2,k}=1,\quad
  \sum_{i=1}^3\sum_{k=1}^2 x_{i,3,k}=1;\\
\text{(each item once)}\;\;
&\sum_{j=1}^3\sum_{k=1}^2 x_{1,j,k}=1,\quad
  \sum_{j=1}^3\sum_{k=1}^2 x_{2,j,k}=1,\quad
  \sum_{j=1}^3\sum_{k=1}^2 x_{3,j,k}=1;\\
\text{(capacity, soft)}\;\;
&\sum_{j=1}^3\sum_{i=1}^3 d_i\,x_{i,j,1}\le Q,\qquad
  \sum_{j=1}^3\sum_{i=1}^3 d_i\,x_{i,j,2}\le Q,
\end{align}
with soft penalty weight \(\lambda_{\mathrm{cap}}\) or enforced as a hard feasibility filter.

\emph{Objective (diagonal).}
\[
\begin{aligned}
H_{\mathrm{obj}}
&=\sum_{k=1}^2\Bigl[\;\sum_{i} w_{\mathrm{dep}\to i}\,x_{i,1,k}
+\sum_{j=1}^{2}\sum_{i,i'} W_{i,i'}\,x_{i,j,k}\,x_{i',j+1,k}
+\sum_{i} w_{i\to\mathrm{dep}}\,x_{i,3,k}\Bigr]\\
&\quad
+\sum_{j=1}^{2}\sum_{k\ne k'}\sum_{i,i'}
\bigl(w_{i\to\mathrm{dep}}+w_{\mathrm{dep}\to i'}\bigr)\,x_{i,j,k}\,x_{i',j+1,k'}.
\end{aligned}
\]

Index customers as $i=1,2,3$ for nodes $2,3,4$ respectively.
The pairwise customer distance matrix $W\in\mathbb R^{3\times 3}$ and depot legs are
\[
W \!=\!
\begin{bmatrix}
0 & 30.41 & 36.40\\
30.41 & 0 & 6.08\\
36.40 & 6.08 & 0
\end{bmatrix},
\]

\[
d(\mathrm{dep}\!\to\! i) \!=\! 
\begin{bmatrix}25.55 & 26.02 & 30.02\end{bmatrix},
\ 
d(i\!\to\!\mathrm{dep}) \!=\!
\begin{bmatrix}25.55 & 26.02 & 30.02\end{bmatrix}.
\]

Pick the feasible colored assignment “truck~1 takes customer~1 at position~1; truck~2 takes
customers~2,3 at positions~2,3”. 

\paragraph{Encoded bitstring.}
The assignment amounts to
\[
j{=}1:(i{=}1,k{=}1),\quad
j{=}2:(i{=}2,k{=}2),\quad
j{=}3:(i{=}3,k{=}2)
\]
which can be encoded by the three one–hot blocks (one per position):
\[
\begin{aligned}
\text{pos }1 &: [1,0,0,0,0,0],\\
\text{pos }2 &: [0,0,0,0,1,0],\\
\text{pos }3 &: [0,0,0,0,0,1].
\end{aligned}
\]
Concatenating these 6-bit blocks (for $n=3$ positions) gives the 18-bit string
\[
\boxed{100000\;000010\;000001}
\]

The per-truck slices (rows=positions $t$, cols=customers $i$) are:
\[
M^{(1)} \;=\;
\begin{bmatrix}
1&0&0\\[2pt]
0&0&0\\[2pt]
0&0&0
\end{bmatrix},
\qquad
M^{(2)} \;=\;
\begin{bmatrix}
0&0&0\\[2pt]
0&1&0\\[2pt]
0&0&1
\end{bmatrix}.
\]
Summing (and transposing) gives the permutation matrix
\[
P \;=\; \Bigl(M^{(1)}+M^{(2)}\Bigr)^{\!\top}
\;=\;
\begin{bmatrix}
1&0&0\\[2pt]
0&1&0\\[2pt]
0&0&1
\end{bmatrix},
\]
which encodes the order “\(\text{pos }1\!\mapsto\! i{=}1,\ \text{pos }2\!\mapsto\! i{=}2,\ \text{pos }3\!\mapsto\! i{=}3\)”.

\subsection{Example B: 5 locations 2 Vehicles}
Similarly consider an example with  \(\boldsymbol{5}\) locationsand  \(\boldsymbol{2}\) vehicles ( \(\Rightarrow n=4\) customers and \(K=2\)).
Blocks \(m=4\). Local alphabet size \(nK=4\cdot2=8\). Total variables/qubits \(n^2K=32\).
With the same \(\texttt{CAPACITY}=3\) and \(d_i\equiv 1\), then a single truck cannot legally serve all
4 customers (load \(4>3\)). The capacity penalty (or a hard filter) therefore discourages
“one truck does everything” assignments and pushes the solution toward a split (e.g.\ \(2{+}2\) or \(3{+}1\)).

\bigskip
\emph{Variables/constraints.} As above, now with \(i,j\in\{1,2,3,4\}\), \(k\in\{1,2\}\).
The same two universal constraints in \eqref{eq:const} make
\(\sum_k x_{i,j,k}\) a \(4\times4\) permutation matrix and the third makes capacity \(\sum_{j=1}^4\sum_{i=1}^4 d_i\,x_{i,j,k}\le 3\) for each \(k\).

\noindent
Index customers as $i=1,2,3,4$ for nodes $2,3,4,5$.
The distance matrix and depot legs are
\[
W \!=\!
\begin{bmatrix}
0 & 43.01 & 37.07 & 23.09\\
43.01 & 0 & 30.41 & 59.08\\
37.07 & 30.41 & 0 & 60.01\\
23.09 & 59.08 & 60.01 & 0
\end{bmatrix},
\]
\[
d(\mathrm{dep}\!\to\! i) \!=\!
\begin{bmatrix}59.03 & 53.04 & 78.10 & 61.41\end{bmatrix},
\;
d(i\!\to\!\mathrm{dep}) \!=\!
\begin{bmatrix}59.03 & 53.04 & 78.10 & 61.41\end{bmatrix}.
\]

Pick the feasible colored assignment with order
\(\text{pos }1\!\mapsto\! i{=}3,\ \text{pos }2\!\mapsto\! i{=}1,\ \text{pos }3\!\mapsto\! i{=}4,\ \text{pos }4\!\mapsto\! i{=}2\),
and truck labels per position \([2,2,1,1]\) (so truck~2 takes the first two, truck~1 the last two). i.e.\ customers $i=1,\ldots,4$ grouped by vehicle label $k=1,2$, with the following assignment
\[
j{=}1:(i{=}3,k{=}2),\quad
j{=}2:(i{=}1,k{=}2),\quad
j{=}3:(i{=}4,k{=}1),\quad
j{=}4:(i{=}2,k{=}1)
\]
is encoded by the following one–hot blocks:
\[
\begin{aligned}
\text{pos }1 &: [0,0,0,0,0,0,1,0],\\
\text{pos }2 &: [0,0,0,0,1,0,0,0],\\
\text{pos }3 &: [0,0,0,1,0,0,0,0],\\
\text{pos }4 &: [0,1,0,0,0,0,0,0].
\end{aligned}
\]
This is equivalent to the 32--bit string
\[
\boxed{
00000010\;00001000\;00010000\;01000000
}
\]

Each 8--bit segment corresponds to one column of the global
assignment tensor $M3$, whose slices $M^{(1)}$ and $M^{(2)}$
reconstruct the permutation matrices in the example above.

The slices are
\[
M^{(1)} \;=\;
\begin{bmatrix}
0&0&0&0\\[2pt]
0&0&0&0\\[2pt]
0&0&0&1\\[2pt]
0&1&0&0
\end{bmatrix},
\qquad
M^{(2)} \;=\;
\begin{bmatrix}
0&0&1&0\\[2pt]
1&0&0&0\\[2pt]
0&0&0&0\\[2pt]
0&0&0&0
\end{bmatrix}.
\]
Hence
\[
P \;=\; \Bigl(M^{(1)}+M^{(2)}\Bigr)^{\!\top}
\;=\;
\begin{bmatrix}
0&1&0&0\\[2pt]
0&0&0&1\\[2pt]
1&0&0&0\\[2pt]
0&0&1&0
\end{bmatrix},
\]
which is a permutation matrix. Each column and each row sums to \(1\), and the per-vehicle slices
\(M^{(k)}\) are disjoint \(0/1\) partial permutation matrices whose sum, after a transpose, gives the customer\(\times\)position permutation matrix \(P\).

\section{Mechanism behind the inherited Fej\'er success bound}
\label{app:fejer-mechanism}

For completeness, we briefly recall the analytic mechanism behind the
dimension-free finite-depth and finite-shot success bound inherited from
Ref.~\cite{onahfinite}. The discussion is given in the notation of the present
paper and may be read as a reference model for the routing specialization. We work on the encoded one-hot manifold
\[
\OH \;:=\; (\mathcal H_1)^{\otimes m},
\qquad
\mathcal H_1 \;:=\; \mathrm{span}\{\ket{e_1},\dots,\ket{e_{n_{\mathrm{loc}}}}\},
\]
where $m$ is the number of blocks and $n_{\mathrm{loc}}$ is the local alphabet
size. In the CVRP construction of the main text one has
$m=n$ and $n_{\mathrm{loc}}=nK$, while in the PDP construction one has
$m=T$ and $n_{\mathrm{loc}}=TK$. The mixer is the block-local XY unitary
\[
U_M(\beta)
\;:=\;
\bigotimes_{b=1}^{m} e^{-i\beta \widetilde H^{(b)}_{XY}},
\]
and the initial state is the uniform one-hot product state
\[
\ket{s_0}
\;:=\;
\ket{s_{\mathrm{blk}}}^{\otimes m},
\qquad
\ket{s_{\mathrm{blk}}}
\;:=\;
\frac{1}{\sqrt{n_{\mathrm{loc}}}}
\sum_{j=1}^{n_{\mathrm{loc}}}\ket{e_j}.
\]
The depth-$p$ CE--QAOA state is
\begin{equation}
\label{eq:app-ce-state}
\ket{\psi_p}
\;=\;
\Bigl(\prod_{r=1}^{p} U_M(\beta_r)\,U_C(\gamma_r)\Bigr)\ket{s_0},
\qquad
U_C(\gamma)\;:=\;e^{-i\gamma H_C},
\end{equation}
where $H_C$ is diagonal in the computational basis
$\{\ket{z}\}_{z\in\OH}$, so that
\[
H_C\ket{z}=E(z)\ket{z}.
\]

The success theorem of Ref.~\cite{onahfinite} is formulated in a
lattice-normalized phase picture. Accordingly, we assume that after a known
global rescaling of $H_C$, the relevant energies admit a phase
representation on the circle with controlled resolution. Fix a base angle
$\gamma$ and define the wrapped phase
\[
\theta(z)\;:=\;\gamma E(z)\ \ (\mathrm{mod}\ 2\pi).
\]
Let $\Omega^\star\subseteq\OH$ denote the set of optimal basis strings, and let
$\theta^\star$ be the common wrapped phase of the optimal energy level. The
wrapped phase separation is
\begin{equation}
\label{eq:app-phase-gap}
\delta
\;:=\;
\min_{y\notin\Omega^\star}
\operatorname{dist}_{\mathbb T}\!\bigl(\theta(y),\theta^\star\bigr),
\end{equation}
where
\[
\operatorname{dist}_{\mathbb T}(\phi,\varphi)
\;:=\;
\min_{k\in\mathbb Z}|\phi-\varphi+2\pi k|
\;\in\;[0,\pi].
\]
The main theorem applies when $\delta>0$. As discussed in
Ref. \cite{onahfinite}, this exact phase-gap assumption may be relaxed by averaging arguments when several near-optimal levels cluster in
phase.

The mechanism is most transparent in a dephased reference model. Let
\[
\mathcal T(\rho)
\;:=\;
\int_{0}^{2\pi}\frac{d\phi}{2\pi}\,
e^{-i\phi H_C}\rho\,e^{+i\phi H_C}
\]
be the cost-basis dephasing channel. This channel removes coherences between
distinct $H_C$ eigenspaces and leaves only the diagonal occupation weights.
We use it only as an analytic baseline that exposes a positive trigonometric
filter. We do not claim that the resulting surrogate dynamics is identical to
the fully coherent CE--QAOA circuit.

After inserting $\mathcal T$ after each layer, the diagonal evolves by a
classical Markov update. Writing
$v^{(r)}(z):=\bra{z}\rho_r\ket{z}$ for the diagonal distribution after $r$
layers, one finds
\[
v^{(r)}(z)
\;=\;
\sum_{y\in\OH} M_{\beta_r}(z|y)\,v^{(r-1)}(y),
\qquad
M_{\beta}(z|y)
\;:=\;
\bigl|\bra{z}U_M(\beta)\ket{y}\bigr|^2.
\]
Each $M_\beta$ is unistochastic and hence doubly stochastic. Iterating from
the initial diagonal
\[
v^{(0)}(z)=|\braket{z}{s_0}|^2
\]
gives the dephased mixer envelope
\begin{equation}
\label{eq:app-Wp}
W_p(z;\boldsymbol\beta)
\;:=\;
\bigl[M_{\beta_p}\cdots M_{\beta_1}v^{(0)}\bigr](z).
\end{equation}
This quantity depends only on the mixer schedule and the encoded initial
state. It is nonnegative and normalized on $\OH$. To expose the cost-phase dependence, one introduces the normalized Dirichlet
polynomial
\begin{equation}
\label{eq:app-dirichlet}
D_p(H_C)
\;:=\;
\frac{1}{\sqrt{p+1}}
\sum_{r=0}^{p} e^{-ir\gamma H_C}.
\end{equation}
On a computational-basis eigenstate $\ket{z}$ this acts by multiplication with
\[
\frac{1}{\sqrt{p+1}}
\sum_{r=0}^{p} e^{-ir\theta(z)}.
\]
The corresponding positive spectral weight is the Fej\'er kernel
\begin{equation}
\label{eq:app-fejer}
F_p\!\bigl(\theta(z)-\theta^\star\bigr)
\;:=\;
\frac{1}{p+1}
\left|
\sum_{r=0}^{p} e^{ir(\theta(z)-\theta^\star)}
\right|^2
\;=\;
\frac{1}{p+1}
\left(
\frac{
\sin\!\bigl(\frac{(p+1)(\theta(z)-\theta^\star)}{2}\bigr)
}{
\sin\!\bigl(\frac{\theta(z)-\theta^\star}{2}\bigr)
}
\right)^2.
\end{equation}
Thus $F_p$ is a nonnegative trigonometric polynomial that peaks at the target
phase and suppresses off-peak phases.

The dephased reference distribution is obtained by applying the positive
filter to a state whose diagonal is the envelope $W_p(\cdot;\boldsymbol\beta)$.
Equivalently, one may define it directly by the normalized factorized law
\begin{equation}
\label{eq:app-factorization}
\Pr_p^{\mathrm{ref}}[z]
\;:=\;
\frac{
W_p(z;\boldsymbol\beta)\,
F_p(\theta(z)-\theta^\star)
}{
\sum_{y\in\OH}
W_p(y;\boldsymbol\beta)\,
F_p(\theta(y)-\theta^\star)
}.
\end{equation}
This equation isolates the two ingredients that drive the inherited success
bound. The envelope $W_p$ captures how the mixer spreads mass across the
encoded manifold, while the Fej\'er factor favors strings whose cost phase is
close to the optimal phase.

It is convenient to define the optimal-set envelope weight
\begin{equation}
\label{eq:app-Cbeta}
C_\beta
\;:=\;
\sum_{x\in\Omega^\star} W_p(x;\boldsymbol\beta).
\end{equation}
Since $F_p(0)=p+1$, the total probability of the optimal set under the
reference law is bounded below by a competition between the on-peak mass
$(p+1)C_\beta$ and the off-peak contribution of all nonoptimal strings.
The relevant off-peak quantity is
\begin{equation}
\label{eq:app-Mp}
M_p(\delta)
\;:=\;
\max_{|\vartheta|\ge \delta} F_p(\vartheta).
\end{equation}
Using the standard Fej\'er estimate, one has
\begin{equation}
\label{eq:app-Mp-bound}
M_p(\delta)
\;\le\;
\frac{1}{(p+1)\sin^2(\delta/2)}.
\end{equation}

The mechanism summarized above leads directly to the inherited success bound.

\begin{theorem}[Reference Fej\'er factorization and success bound]
\label{thm:app-fejer-success}
Assume that the routing specialization admits the factorized reference law
\eqref{eq:app-factorization}, and assume a positive wrapped phase separation
$\delta>0$ in the sense of \eqref{eq:app-phase-gap}. Then the success
probability of sampling an optimal basis string satisfies
\begin{equation}
\label{eq:app-q0}
q_0
\;:=\;
\Pr_p^{\mathrm{ref}}[\Omega^\star]
\;=\;
\sum_{x\in\Omega^\star}\Pr_p^{\mathrm{ref}}[x]
\;\ge\;
\frac{(p+1)C_\beta}
{(p+1)C_\beta + M_p(\delta)(1-C_\beta)}.
\end{equation}
In particular, the lower bound is dimension-free. It depends only on the
filter order $p$, the phase gap $\delta$, and the optimal-set envelope weight
$C_\beta$.
\end{theorem}

\begin{proof}
Summing \eqref{eq:app-factorization} over $x\in\Omega^\star$ and using
$F_p(0)=p+1$ gives
\[
q_0
=
\frac{(p+1)C_\beta}
{\sum_{y\in\OH}W_p(y;\boldsymbol\beta)\,F_p(\theta(y)-\theta^\star)}.
\]
For every $y\notin\Omega^\star$, the phase-gap assumption implies
$F_p(\theta(y)-\theta^\star)\le M_p(\delta)$. Hence the denominator is at most
\[
(p+1)C_\beta + M_p(\delta)\sum_{y\notin\Omega^\star}W_p(y;\boldsymbol\beta)
=
(p+1)C_\beta + M_p(\delta)(1-C_\beta),
\]
which yields \eqref{eq:app-q0}.
\end{proof}

Equation~\eqref{eq:app-q0} immediately yields the finite-depth and finite-shot
corollaries used in the main text. If one requires $q_0\ge 1-\varepsilon$, a
sufficient condition is
\[
(p+1)^2
\;\ge\;
\frac{1-\varepsilon}{\varepsilon}\cdot
\frac{1-C_\beta}{C_\beta}\cdot
\csc^2(\delta/2).
\]
Thus the filter order needed to peak on the optimum is finite whenever the
phase gap is positive. Likewise, introducing
\[
A:=(p+1)^2\sin^2(\delta/2),
\qquad
x:=A\,C_\beta,
\]
one obtains
\[
q_0 \ge \frac{x}{(1-C_\beta)+x}\ge \frac{x}{1+x},
\]
and therefore the standard one-hit estimate
\[
S
\;\ge\;
\frac{1}{q_0}\ln\frac{1}{\epsilon}
\;\le\;
\left(1+\frac{1}{x}\right)\ln\frac{1}{\epsilon}.
\]
Hence, once $(p+1)^2\sin^2(\delta/2)C_\beta$ is bounded below by an inverse
polynomial, polynomially many shots suffice to recover an optimum with high
confidence. Finally, we stress again that the dephased Fej\'er construction is used here
as an analytic reference mechanism. Its role is to isolate a positive
trigonometric filter on cost phases and make explicit how mixer spreading and
phase isolation combine to produce a dimension-free optimality bound. The
coherent circuit may further reshape this weight through interference, but the
reference model already identifies the key instance-dependent parameters that
govern success.

\subsection{LP-anchored mixer-angle selection for routing}
\label{subsubsec:lp-anchored-beta}

For routing problems, classical relaxations provide a natural polynomial-time
anchor for selecting mixer angles.  Vehicle-routing formulations and their
linear relaxations are central objects in classical exact and heuristic
methods, including set-partitioning formulations, column generation, and
branch-price-and-cut algorithms
\cite{toth2014vehicle,costa2019branchpricecut}.  Let \(x^{\mathrm{LP}}\)
denote the solution of a linear relaxation of the routing formulation.  Its
entries may represent fractional assignment variables, fractional arc
variables, or fractional customer-position indicators.  The relaxation
provides directional information about promising regions of the encoded space.

Let \(\mathcal A\) denote a set of elementary routing features, such as
assignment features, arc features, or adjacent-position features.  For a
configuration \(z\in\Omega\), let \(\chi_a(z)\in\{0,1\}\) indicate whether
feature \(a\in\mathcal A\) is present in \(z\).  Define the LP-alignment score
\[
        S_{\mathrm{LP}}(\beta)
        :=
        \sum_{a\in\mathcal A}
        x_a^{\mathrm{LP}}\,
        \mathbb E_{z\sim W_p(\cdot;\beta)}[\chi_a(z)].
\]
Equivalently,
\[
        S_{\mathrm{LP}}(\beta)
        =
        \mathbb E_{z\sim W_p(\cdot;\beta)}
        \left[
        \sum_{a\in\mathcal A}x_a^{\mathrm{LP}}\chi_a(z)
        \right].
\]
For pair or arc features, one may use
\[
        S_{\mathrm{LP}}^{(2)}(\beta)
        :=
        \sum_{a,b\in\mathcal A}
        x_{ab}^{\mathrm{LP}}\,
        \mathbb E_{z\sim W_p(\cdot;\beta)}
        [\chi_{ab}(z)],
\]
where \(\chi_{ab}(z)\) indicates whether the pair or arc feature \((a,b)\)
appears in \(z\).

A combined relaxation-energy surrogate is
\[
        S_{\mathrm{route}}(\beta)
        :=
        S_{\mathrm{LP}}^{(2)}(\beta)
        -
        \eta\,\mu_\beta
        -
        \rho\,\sigma_\beta ,
        \qquad
        \eta,\rho\ge 0.
\]
The corresponding angle selection rule is
\[
        \beta_{\mathrm{LP}}
        \in
        \arg\max_{\beta\in\mathcal B}
        S_{\mathrm{route}}(\beta).
\]
This criterion selects mixer angles whose envelope aligns with the fractional
routing relaxation while also favoring low expected cost and controlled
variance.

\begin{definition}[LP-anchored mixer-angle surrogate]
\label{def:lp-anchored-surrogate}
Let \(x^{\mathrm{LP}}\) be a fractional solution of a polynomial-time
relaxation of the routing instance.  Let \(\chi_a\) and \(\chi_{ab}\) denote
encoded routing features.  The LP-anchored mixer score is
\[
        S_{\mathrm{LP}}(\beta)
        =
        \sum_{a}x_a^{\mathrm{LP}}\,
        \mathbb E_{W_p(\cdot;\beta)}[\chi_a],
\]
or, for pair features,
\[
        S_{\mathrm{LP}}^{(2)}(\beta)
        =
        \sum_{a,b}x_{ab}^{\mathrm{LP}}\,
        \mathbb E_{W_p(\cdot;\beta)}[\chi_{ab}].
\]
A relaxation-energy mixer score is
\[
        S_{\mathrm{route}}(\beta)
        =
        S_{\mathrm{LP}}^{(2)}(\beta)
        -
        \eta\mu_\beta
        -
        \rho\sigma_\beta .
\]
\end{definition}

This surrogate is well matched to vehicle-routing formulations because
fractional relaxations often expose useful arc, route, assignment, or
customer-position structure before integrality is enforced
\cite{toth2014vehicle,costa2019branchpricecut}.  The LP anchor guides the
choice of \(\beta\) toward regions of the encoded manifold already favored by the classical relaxation.

The preceding constructions can be combined into a single computable
preselection objective.  A general-purpose surrogate is
\[
        \mathcal S(\beta)
        :=
        \log
        \mathbb E_{z\sim W_p(\cdot;\beta)}
        \bigl[e^{-\lambda C(z)}\bigr]
        +
        \alpha S_{\mathrm{LP}}^{(2)}(\beta)
        -
        \rho\sigma_\beta
        +
        \nu \widetilde C_\beta ,
\]
with tunable weights \(\alpha,\rho,\nu\ge 0\).  The selected mixer angle is
\[
        \beta^\star
        \in
        \arg\max_{\beta\in\mathcal B}
        \mathcal S(\beta).
\]
For general constrained problems, one may omit the LP term and use
\[
        \mathcal S_{\mathrm{gen}}(\beta)
        :=
        \log Z_\lambda(\beta)
        -
        \rho\sigma_\beta .
\]
For routing instances, a natural choice is
\[
        \mathcal S_{\mathrm{route}}(\beta)
        :=
        \log Z_\lambda(\beta)
        +
        \alpha S_{\mathrm{LP}}^{(2)}(\beta)
        -
        \rho\sigma_\beta .
\]

\begin{algorithm}[t]
\caption{Surrogate-based mixer-angle preselection}
\label{alg:surrogate-beta-selection}
\begin{algorithmic}[1]
\Require Encoded instance \(I\), objective \(C\), mixer-envelope family
\(W_p(\cdot;\beta)\), angle search domain \(\mathcal B\), inverse temperature
\(\lambda>0\), optional threshold \(\tau\), optional LP relaxation
\(x^{\mathrm{LP}}\).
\Ensure Mixer angle \(\beta^\star\).

\State Compute cheap classical information:
\[
        \tau \leftarrow \text{greedy/local-search/LP-rounded incumbent}
\]
and, when available,
\[
        x^{\mathrm{LP}} \leftarrow \text{solution of a polynomial relaxation}.
\]

\For{\(\beta\in\mathcal B\)}
    \State Estimate or compute
    \[
        \mu_\beta
        =
        \mathbb E_{W_p(\cdot;\beta)}[C],
        \qquad
        \sigma_\beta^2
        =
        \operatorname{Var}_{W_p(\cdot;\beta)}(C).
    \]
    \State Estimate or compute the soft low-energy score
    \[
        Z_\lambda(\beta)
        =
        \mathbb E_{W_p(\cdot;\beta)}[e^{-\lambda C}].
    \]
    \If{LP relaxation information is available}
        \State Estimate or compute
        \[
            S_{\mathrm{LP}}^{(2)}(\beta)
            =
            \sum_{a,b}x_{ab}^{\mathrm{LP}}
            \mathbb E_{W_p(\cdot;\beta)}[\chi_{ab}].
        \]
    \Else
        \State Set \(S_{\mathrm{LP}}^{(2)}(\beta)=0\).
    \EndIf
    \State Form the combined surrogate
    \[
        \mathcal S(\beta)
        =
        \log Z_\lambda(\beta)
        +
        \alpha S_{\mathrm{LP}}^{(2)}(\beta)
        -
        \rho\sigma_\beta .
    \]
\EndFor

\State Return
\[
        \beta^\star
        \in
        \arg\max_{\beta\in\mathcal B}
        \mathcal S(\beta).
\]
\end{algorithmic}
\end{algorithm}

\subsection{Simulation Details}
\label{app:impl}
The one-hot manifold is prepared by a multi-block encoder
\texttt{multi\_block\_encoder(block\_sz, m\_blocks)}, yielding the product start state
$\ket{s_0}=\ket{s_{\mathrm{blk}}}^{\otimes n}$, where each block is the uniform one-hot superposition over the $S=nK$ symbols. For each instance we construct a QUBO for the global-position formulation using
\texttt{create\_cvrp\_qubo\_global\_positions}. The QUBO includes:
(i) the global each-customer-once constraint as a squared affine penalty,
(ii) a capacity penalty term controlled by \texttt{lam\_cap}, and
(iii) the routing objective that couples consecutive positions on the global timeline and
charges depot-close/depot-open terms at vehicle-switch boundaries as in
Eq.~\eqref{eq:cvrp-A-obj}. We use fixed weights
\[
A_{\text{pen}}=4.0,\qquad B_{\text{obj}}=1.0,\qquad \lambda_{\text{cap}}=A_{\text{pen}},
\]
so the diagonal cost is
\[
H_C \;=\; H_{\mathrm{pen}} + H_{\mathrm{obj}},
\]
and the corresponding Pauli operator is generated via
\texttt{build\_pauli\_from\_qubo}, using the standard Ising map \(x\mapsto (1-Z)/2\).
Although \(H_C\) is diagonal in the computational basis, it contains extensive cross-block
couplings, notably the assignment penalty and the adjacent-position routing terms, so the
phase separator \(e^{-i\gamma H_C}\) is entangling across blocks.

\paragraph{Mixer.}
We implement a tensor product of identical block-local XY mixers:
\[
U_M(\beta)\;=\;\bigotimes_{j=1}^{n} U^{(j)}_M(\beta),
\qquad
U_M^{(j)}(\beta)=e^{-i\beta\,\widetilde H_{XY}^{(j)}}.
\]
In code this is realized by appending \texttt{xy\_mixer\_block(block\_sz, $\beta$)}
to each contiguous register segment of length $S=nK$. This mixer preserves the one-hot
subspace \emph{by construction}, ensuring that the per-position one-hot constraint is
satisfied throughout the evolution (ideal kernel setting). This is crucial for ensuring
that the algorithm samples from the structured encoded space rather than from the full
$2^{n^2K}$ bitstring space.

\paragraph{Ansatz depth and parameterization.}
All experiments use depth $p=1$:
\[
\ket{\psi(\gamma,\beta)}
\;=\;
U_M(\beta)\,e^{-i\gamma H_C}\,\ket{s_0}.
\]
We build the circuit using \texttt{QAOAAnsatz} with
\texttt{initial\_state=init\_state} and \texttt{mixer\_operator=mixer}, and we explicitly
identify the $\gamma$ and $\beta$ parameters robustly by name-matching (``gamma'', ``$\gamma$'', ``beta'', ``$\beta$'') rather than by index. 

\paragraph{Grid search.}
For each instance, we perform a dense uniform grid search over $(\gamma,\beta)$:
\[
\gamma \in \mathrm{linspace}(0,\pi,S+1),\qquad
\beta \in \mathrm{linspace}(0,\pi,S+1),
\]
where $S=nK$ is the block size. Thus, the grid has $(S+1)^2$ points, scaling as
$O((nK)^2)$. This specific choice ties the angular resolution to the local alphabet size
and matches the empirical ``kernel resolution'' used in our CE--QAOA studies.

\paragraph{Shot budget.}
At each grid point we sample with a shots budget
\[
S_{\text{shots}} \;=\; (nK)^3.
\]
We report per grid point the fraction of feasible outputs whose observed frequency exceeds this baseline (\texttt{share\_above\_baseline}) as a diagnostic for encoded anticoncentration relative to the uniform distribution on the encoded space.

\paragraph{Backend and simulation model.}
All circuits are executed using Qiskit Aer in MPS mode. This choice allows simulation of moderately sized circuits beyond exact statevector limits,
while controlling numerical cost via truncation. The bond-dimension cap (128) and truncation
thresholds (typically $10^{-3}$) were chosen to balance runtime and numerical stability.
We emphasize that these truncations are \emph{part of the simulator model} and therefore
the reported results reflect an MPS-approximated circuit execution.

\paragraph{Transpilation strategy: translator path with HLS disabled.}
To ensure consistent circuit synthesis and to avoid high-level synthesis heuristics that may
introduce variability across runs, we force the \texttt{translator} translation method and
disable HLS by setting an empty \texttt{HLSConfig(op\_types=[])}. We transpile into the fixed
basis gate set:
\[
\{\texttt{id},\ \texttt{rz},\ \texttt{sx},\ \texttt{x},\ \texttt{cx},\ \texttt{measure},\ \texttt{barrier}\},
\]
at optimization level 3. Importantly, we transpile the parameterized template circuit
\emph{once} and then perform fast parameter binding inside the grid loop:
\[
\texttt{template\_tr.assign\_parameters(\{gamma:..., beta:...\})}.
\]
This avoids repeated compilation overhead and ensures that all grid points correspond to the
same compiled circuit topology with only rotation angles changed.

\paragraph{Sampling primitive and frequency reconstruction.}
We use \texttt{BackendSampler} for circuit execution. The sampler returns quasi-distributions
over bitstrings. When raw shot-count dictionaries are unavailable, we convert these to
frequency proxies by
\[
\nu(b) \;=\; p(b)\,S_{\text{shots}},
\]
and retain only outcomes with \(\nu(b)>0\). We emphasize that \(\nu(b)\) is then a weighted
frequency proxy rather than an exact hardware count. All multiplicity-based diagnostics in this
paper are computed from these frequency proxies. Because key formats can be integer-encoded or
string-encoded depending on backend and settings, we normalize keys to fixed-width
bitstrings via a helper (\texttt{\_key\_to\_bitstr}) that removes spaces and applies
zero-padding to width \(Q=n^2K\). See implementation in Ref. \cite{onahcvrpdata}.

\paragraph{Exact checker.}
For each feasible bitstring \(b\), evaluates its exact routing objective
\[
E(b)=\langle b \mid H_{\mathrm{obj}} \mid b\rangle
\]
(or, if desired, the full diagonal score \(\langle b\mid H_C\mid b\rangle\)). PHQC then returns the feasible sample of minimum score among all samples collected across the full grid \(\mathcal G\). Thus PHQC succeeds as soon as an optimal feasible bitstring is observed at least once.

\begin{definition}[Optimal feasible set at a given instance]
Let
\[
\mathcal X^\star
:=
\bigl\{
b \in \{0,1\}^{n^2K} :
\textsc{FeasibleGlobalPositions}(b)=\texttt{true},
\ \ E(b)=E^\star
\bigr\},
\]
where
\[
E^\star
=
\min\{
E(b):
\textsc{FeasibleGlobalPositions}(b)=\texttt{true}
\}.
\]
For each parameter pair \((\beta,\gamma)\), define the optimal feasible mass
\[
p_\star(\beta,\gamma)
:=
\Pr_{\psi_p(\beta,\gamma)}[\,b\in\mathcal X^\star\,].
\]
\end{definition}

\begin{lemma}[No-hit bound at a good grid point]
\label{lem:phqc-nohit}
Assume there exists a parameter pair
\((\beta^\sharp,\gamma^\sharp)\in\mathcal G\)
such that
\[
p_\star^\sharp
:=
p_\star(\beta^\sharp,\gamma^\sharp)
>0.
\]
If PHQC draws \(S\) independent samples at each grid point, then the probability of observing no optimal feasible bitstring at the good grid point satisfies
\[
\Pr[\text{no hit at }(\beta^\sharp,\gamma^\sharp)]
=
(1-p_\star^\sharp)^S
\le
e^{-p_\star^\sharp S}.
\]
Hence
\[
S \ge \Bigl\lceil \frac{\ln(1/\delta)}{p_\star^\sharp}\Bigr\rceil
\quad\Longrightarrow\quad
\Pr[\text{at least one optimal feasible hit at }(\beta^\sharp,\gamma^\sharp)]
\ge 1-\delta.
\]
\end{lemma}

\begin{proof}
Let \(X_r\sim\mathrm{Bernoulli}(p_\star^\sharp)\) indicate whether shot \(r\) at
\((\beta^\sharp,\gamma^\sharp)\) lands in \(\mathcal X^\star\). Then
\[
N_\star=\sum_{r=1}^{S} X_r
\]
counts optimal feasible hits at that grid point. Therefore
\[
\Pr[N_\star=0]=(1-p_\star^\sharp)^S \le e^{-p_\star^\sharp S},
\]
which gives the claim.
\end{proof}

\begin{theorem}[Conditional exact-recovery guarantee for PHQC]
\label{thm:phqc-exact}
Assume the coarse grid \(\mathcal G\) contains a parameter pair
\((\beta^\sharp,\gamma^\sharp)\) with optimal feasible mass
\(p_\star^\sharp>0\). Run PHQC with \(S\) shots at each grid point and let the deterministic checker return the minimum-cost feasible bitstring over all samples. Then
\[
S \ge \Bigl\lceil \frac{\ln(1/\delta)}{p_\star^\sharp}\Bigr\rceil
\quad\Longrightarrow\quad
\Pr[\text{PHQC returns a globally optimal feasible bitstring}]
\ge 1-\delta.
\]
In particular, if \(p_\star^\sharp \ge n^{-k}\) for some constant \(k>0\), then
\[
S = O\!\bigl(n^k\log(1/\delta)\bigr),
\]
and if moreover \(|\mathcal G|=\mathrm{poly}(n)\), the full quantum-plus-classical pipeline has polynomial shot complexity and polynomial total runtime.
\end{theorem}

\begin{proof}
Since PHQC samples every parameter pair in \(\mathcal G\), it samples the good grid point \((\beta^\sharp,\gamma^\sharp)\). By Lemma~\ref{lem:phqc-nohit}, with the stated shot budget the probability of observing at least one bitstring in \(\mathcal X^\star\) at that point is at least \(1-\delta\). The deterministic checker filters infeasible outcomes and returns the minimum-cost feasible bitstring among all observed samples. Therefore, once any bitstring from \(\mathcal X^\star\) appears, PHQC returns the optimal solution.
\end{proof}

\printbibliography

\end{document}